\documentclass[a4paper,twoside,english,nofootinbib, preprint]{revtex4}
\usepackage{palatino}
\usepackage[T1]{fontenc}
\usepackage[latin1]{inputenc}
\usepackage{calc}
\usepackage{amsmath}
\usepackage{graphicx}
\usepackage{amssymb}

\makeatletter

\newenvironment{lyxlist}[1]
{\begin{list}{}
{\settowidth{\labelwidth}{#1}
 \setlength{\leftmargin}{\labelwidth}
 \addtolength{\leftmargin}{\labelsep}
 }}
{\end{list}}

\usepackage{layouts}
\usepackage{bbm}
\usepackage{hyperref}
\usepackage{mathrsfs}
\usepackage[all]{xy}
\usepackage{feynmf}

\usepackage{amssymb}
\usepackage{amsthm}
\usepackage{bbm}
\usepackage{bm}

\theoremstyle{plain}
  \newtheorem{theorem}{Theorem}[section]
  \newtheorem{proposition}[theorem]{Proposition}  
\theoremstyle{definition}
  \newtheorem{definition}[theorem]{Definition}
  \newtheorem{condition}{Condition}
  \newtheorem*{IndHyp}{Induction hypothesis}
\theoremstyle{remark}
  \newtheorem{remark}[theorem]{Remark}
  \newtheorem{example}[theorem]{Example}

\newcommand{\adv}{\mathrm{adv}}

\newcommand{\D}{\mathscr{D}} 
\newcommand{\dE}{\cdot_{\mathrm{E}}} 
\newcommand{\Diag}{\mathrm{Diag}}

\newcommand{\E}{\mathscr{E}} 
\newcommand{\Em}{\mathbb{E}} 

\newcommand{\FT}{\mathrm{FT}} 

\renewcommand{\L}{\mathscr{L}}
\newcommand{\loc}{\mathrm{loc}}

\newcommand{\ret}{\mathrm{ret}}

\newcommand{\sd}{\mathrm{sd}}

\newcommand{\supp}{\mathrm{supp}} 

\newcommand{\Sym}{\mathrm{Sym}}

\renewcommand{\vec}{\boldsymbol}

\newcommand{\WF}{\mathrm{WF}} 

\usepackage{babel}
\makeatother
\begin{document}

\newcommand{\kk}[1]{\overline{#1}}

\newcommand{\fps}[2]{#1\![[#2]]}

\begin{center}

\par\end{center}

\preprint{\noindent DESY 09-031, ZMP-HH/09-6}

\begin{abstract}
In the framework of perturbative Algebraic Quantum Field Theory (pAQFT)
recently developed by Brunetti, Dütsch, and Fredenhagen (\href{http://arxiv.org/abs/0901.2038}{arXiv:0901.2038}),
I give a general construction of so-called {}``Euclidean time-ordered
products'', i.e. algebraic versions of the Schwinger functions, for
scalar quantum field theories on spaces of Euclidean signature. This
is done by generalizing the recursive construction of time-ordered
products by Epstein and Glaser, originally formulated for quantum
field theories on Minkowski space (MQFT). An essential input of Epstein-Glaser
renormalization is the causal structure of Minkowski space. The absence
of this causal structure in the Euclidean framework makes it necessary
to modify the original construction of Epstein and Glaser at two points.
First, the whole construction has to be performed with an only partially
defined product on (interaction-) functionals. This is due to the
fact that the fundamental solutions of the Helmholtz operator $\left(-\Delta+m^{2}\right)$
of EQFT have a unique singularity structure, i.e. they are unique
up to a smooth part. Second, one needs to (re-)introduce a (rather
natural) {}``Euclidean causality'' condition for the recursion of
Epstein and Glaser to be applicable.
\end{abstract}


\title{Euclidean Epstein-Glaser Renormalization}

\author{Kai J. Keller}

\affiliation{{II.} Institute for Theoretical Physics, Hamburg University,\\
Luruper Chaussee 149, 22761 Hamburg, Germany}

\email{kai.johannes.keller@desy.de}

\maketitle
\tableofcontents{}

\begin{fmffile}{EGeukl28}


\newcommand{\FGH}{
\begin{fmfgraph}(20,15)
\fmfbottom{g,h} \fmftop{f}
\fmfdot{f,g,h}
\end{fmfgraph}
}


\newcommand{\FGoneHone}{
\begin{fmfgraph}(20,15)
\fmfbottom{g,h} \fmftop{f}
\fmfdot{f,g,h}
\fmf{plain}{g,h}
\end{fmfgraph}
}

\newcommand{\FoneGHone}{
\begin{fmfgraph}(20,15)
\fmfbottom{g,h} \fmftop{f}
\fmfdot{f,g,h}
\fmf{plain}{f,h}
\end{fmfgraph}
}

\newcommand{\FoneGoneH}{
\begin{fmfgraph}(20,15)
\fmfbottom{g,h} \fmftop{f}
\fmfdot{f,g,h}
\fmf{plain}{f,g}
\end{fmfgraph}
}


\newcommand{\FGtwoHtwo}{
\begin{fmfgraph}(20,15)
\fmfbottom{g,h} \fmftop{f}
\fmfdot{f,g,h}
\fmf{plain,right=.5}{g,h}
\fmf{plain,left=.5}{g,h}
\end{fmfgraph}
}

\newcommand{\FoneGoneHtwo}{
\begin{fmfgraph}(20,15)
\fmfbottom{g,h} \fmftop{f}
\fmfdot{f,g,h}
\fmf{plain}{f,h}
\fmf{plain}{g,h}
\end{fmfgraph}
}

\newcommand{\FoneGoneoneHone}{
\begin{fmfgraph}(20,15)
\fmfbottom{g,h} \fmftop{f}
\fmfdot{f,g,h}
\fmf{plain}{f,g}
\fmf{plain}{g,h}
\end{fmfgraph}
}

\newcommand{\FtwoGHtwo}{
\begin{fmfgraph}(20,15)
\fmfbottom{g,h} \fmftop{f}
\fmfdot{f,g,h}
\fmf{plain,right=.5}{f,h}
\fmf{plain,left=.5}{f,h}
\end{fmfgraph}
}

\newcommand{\FtwoGoneHone}{
\begin{fmfgraph}(20,15)
\fmfbottom{g,h} \fmftop{f}
\fmfdot{f,g,h}
\fmf{plain}{f,g}
\fmf{plain}{f,h}
\end{fmfgraph}
}

\newcommand{\FtwoGtwoH}{
\begin{fmfgraph}(20,15)
\fmfbottom{g,h} \fmftop{f}
\fmfdot{f,g,h}
\fmf{plain,right=.5}{f,g}
\fmf{plain,left=.5}{f,g}
\end{fmfgraph}
}


\newcommand{\FthreeGoneHtwoF}{
\begin{fmfgraph*}(40,30)
\fmfbottom{g,h} \fmftop{f}
\fmfdot{f,g,h}
\fmflabel{$F^{(5)}$}{f}
\fmflabel{$G^{(4)}$}{g}
\fmflabel{$H^{(3)}$}{h}
\fmf{plain,right=.3}{f,g}
\fmf{plain}{f,g}
\fmf{plain,left=.3}{f,g}
\fmf{plain,right=.3}{g,h}
\fmf{plain,right=.3}{h,f}
\fmf{plain,left=.3}{h,f}
\end{fmfgraph*}
}

\section{Introduction}

In perturbative quantum field theory (pQFT) one is interested in the
terms of the expansion of the $S$-Matrix, i.e. the time-ordered exponential\begin{align}
S(V)=\exp_{\mathrm{T}}(V) & =\sum_{n=0}^{\infty}\frac{1}{n!}V\cdot_{\mathrm{T}}\cdots\cdot_{\mathrm{T}}V=\sum_{n=0}^{\infty}\frac{1}{n!}S^{\left(n\right)}(V^{\otimes n})\,.\label{eq:S-Matrix}\end{align}
$S^{\left(n\right)}$ denotes here the $n^{\mathrm{th}}$ functional
derivative of $S$ with respect to the interaction functional $V:\varphi\rightarrow V(\varphi)$.
As is well known the terms of this expansion, referred to as \emph{time-ordered
products}, give information about transition probabilities in collision
processes of elementary particles (LSZ-relations). The problem occurring
here, referred to as the \emph{renormalization problem} (of pQFT)
is that the time-ordered product $F\cdot_{\mathrm{T}}G$ of two functionals
is generally ill-defined if the supports of the functionals intersect.
The aim of renormalization thus is to make sense of the time-ordered
product also for (local) functionals with coinciding supports.

Although there is a mathematically rigorous formulation of renormalization
on Minkowski-, or even curved Lorentzian spacetimes \citep{Brunetti2000},
it still seems somewhat far from the tools applied in concrete calculations
of transition probabilities, which in turn are known to be in excellent
accordance with experimental data. These calculations are often performed
on spaces of Euclidean signature, which leads to {}``easier'' expressions
and is possible since the fundamental solutions of the Klein-Gordon
operator depend on the hyperbolic distance only (cf. \citep{BolliniGiambiagi1996}).
The way back to MQFT however is not always open (cf. \citep{OsterwalderSchrader1973,OsterwalderSchrader1975,EckmannEpstein1979}).

In the standard approach the passage to Euclidean signature is performed
by an analytic continuation of the Wightman functions%
\footnote{These are the correlation functions of Minkowski QFT, i.e. vacuum
expectation values of products of fields.%
} to the {}``permuted extended tubes'' and evaluation at so-called
Euclidean points or {}``Schwinger points'' \citep{StreaterWightman1964,Schwinger1959,Symanzik1969},
e.g. $\left(ix_{0},\dots,x_{4};iy_{0},\dots,y_{4}\right)$ for the
two-point-function in $D=4$. At these points the \emph{hyperbolic
distance} takes the form of a (negative) \emph{Euclidean distance}:\begin{equation}
x_{0}^{2}-x_{1}^{2}-x_{2}^{2}-x_{3}^{2}\mapsto-x_{0}^{2}-x_{1}^{2}-x_{2}^{2}-x_{3}^{2}.\label{eq:WickRotation}\end{equation}
Because the transition to Euclidean signature amounts to {}``rotating''
the time coordinate by $i=e^{i\frac{\pi}{2}}$ in the complex plane
it is often referred to as {}``Wick rotation''. Performing calculations
using the Wick-rotated Wightman functions has the (rather obvious)
advantage that the Euclidean distance on the right hand side of (\ref{eq:WickRotation})
vanishes only in the origin, whereas the set of zeros of the Minkowskian
distance on the left hand side is the whole forward and backward lightcone.
This entails that in the Wick rotated setting the amplitudes of graphs
with (at most) one loop can be made absolutely convergent. Divergences
of higher loop order and especially so-called overlapping divergences
can then be removed using the graph-by-graph method of Bogoliubov,
Parasiuk, Hepp and Zimmermann, abbreviated BPHZ renormalization \citep{Bogoliubow1957,Hepp1966,Zimmermann1969}.

On the other hand one great disadvantage of the Euclidean framework
becomes apparent at this stage already. The causal structure as in
any relativistic theory encoded very nicely in the Minkowski signature,
is completely lost.

Causality, however is a major ingredient in the formulation of (perturbative)
QFT on Minkowski and curved, globally hyperbolic spacetimes (cf. \citep{StueckelbergRivier1950,BogoliubovShirkov1959,Steinmann1971,Epstein1973};
and \citep{Brunetti2000,HollandsWald2001,Brunetti2009} respectively).
There are two main points where it enters in the formalism. One is
the construction of the algebra of observables, where it enters the
definition of the star-product%
\footnote{The {}``deformation quantizational'' viewpoint has proven to be
both, of structural clarity and convenience for the investigation
of perturbative QFT \citep{Dutsch2000,Hirshfeld2002,Brunetti2009}%
} in form of the causal propagator $E=E^{\ret}-E^{\adv}$ fulfilling\[
\left[\varphi(x),\varphi(y)\right]_{\star}=iE(x-y)\]
in the sense of distribution kernels. The second point is the recursive
construction of time-ordered products in Epstein-Glaser renormalization.
There it enters in the form of the {}``causality condition'', which
makes the construction of time-ordered products up to the thin diagonal
possible.

The aim of this note is to develop a Euclidean version of Epstein-Glaser
renormalization in order to investigate the local (i.e. {}``UV''-)
structure of Euclidean pQFT. In particular we want to gain a deeper
understanding of the relation of the two viewpoints briefly introduced
above, the BPHZ procedure mostly applied in the Euclidean setting
and the Epstein-Glaser recursion, seemingly tied to the causal structure
of spacetimes of Minkowski signature. Besides this there is a second
motivation. The fact that the formulation of Epstein-Glaser renormalization
in the Euclidean framework is possible, despite the absence of a globally
defined star-product, suggests that the whole recursive procedure
of Epstein and Glaser does not depend on the star-product structure
of Minkowskian pQFT at all. Consequently it should be possible to
perform the same construction on Minkowski spacetime using the time-ordered
product only.

As asserted above in this article we are only concerned with the local
properties of the theory, i.e. we do not take vacuum expectation values
of products of fields (Wightman functions) and then perform the Wick
rotation to EQFT with them. We rather regard the implications of going
to Euclidean signature on the algebraic level (i.e. before expectation
values in certain states are taken into account). This already gives
us all the information about the local properties of the theory. Evaluating
the newly defined Euclidean time-ordered products in some {}``vacuum
state'' and performing the adiabatic limit would give us back the
original Schwinger functions. This last step however is not our concern
in the present note.

The strategy of the construction presented here is as follows. We
will use a fundamental solution $P$ of the Helmholtz operator $\left(-\Delta+m^{2}\right)$,
i.e. the {}``Wick rotated'' Klein-Gordon operator, to define a {}``Euclidean
time-ordered product'' for functionals with disjoint supports. These
functionals will then form an associative \emph{partial algebra},
i.e. an algebra with only partially defined product. Due to associativity
the $n$-fold Euclidean time-ordered product can be defined as a multi-linear
map $E_{n}$ on this partial algebra. We then introduce a supplement
for the causality condition of Epstein and Glaser called {}``Euclidean
causality'' which makes it possible to extend the domain of definition
of $E_{n}$ to tensor products of functionals whose support does not
intersect the thin diagonal. The Epstein-Glaser induction closes if
an extension of the domain of definition to the thin diagonal is possible.
As in the original work of Epstein and Glaser the extension problem
can be reduced to the extension of certain scalar distributions, which
in the case of flat (Euclidean) space are translation invariant. This
reduces the extension problem for the $E_{n}$ to that of extending
the domain of the scalar distributions to the origin. The extension
problem for scalar distributions however is well understood \citep{Steinmann1971,Epstein1973}
and is most conveniently discussed in terms of two fundamental theorems
by Brunetti and Fredenhagen \citep[Thm. 5.2 \& 5.3]{Brunetti2000}.

\section{\label{sec:Preliminaries}Preliminaries}

In this section I introduce the basic setup of perturbative Algebraic
Quantum Field Theory (pAQFT) as developed by Brunetti, Dütsch and
Fredenhagen \citep{Brunetti2009} applied, however, to the Euclidean
setting. A remark on how to translate the concepts from Minkowski
to Euclidean signature was also given by R. Stora, see \citep{Stora2006}
for instance.

Within this article let $\Em$ be a $d$-dimensional Euclidean space
and $\mathcal{C}(\Em)\equiv C^{\infty}(\Em)$ the configuration space
of a scalar field theory. Let furthermore $\widetilde{\mathcal{F}}(\Em)$
be the space of \emph{smooth functionals}. These are maps $F:\mathcal{C}(\Em)\rightarrow\mathbb{C}$
for which the $n^{\mathrm{th}}$ functional derivative, denoted by
\[
\left\langle F^{\left(n\right)}(\varphi),h^{\otimes n}\right\rangle \equiv\frac{\delta^{n}F}{\delta\varphi^{n}}(h^{\otimes n}):=\frac{d^{n}}{d\lambda^{n}}\bigg|_{\lambda=0}F(\varphi+\lambda h)\,,\]
exists as a symmetric distribution in $n$ variables, $F^{\left(n\right)}(\varphi)\in\D'(\Em^{n})$.%
\footnote{We will generally assume the functionals occuring in this article
to be smooth in the above sense. %
} We define $\mathcal{F}(\Em)\subset\widetilde{\mathcal{F}}(\Em)$
to be the subspace of (smooth) functionals of compact support, i.e.
for $F\in\mathcal{F}(\Em)$ and all $n\in\mathbb{N}$ the $n^{\mathrm{th}}$
functional derivative $F^{\left(n\right)}(\varphi)\equiv\frac{\delta^{n}F}{\delta\varphi^{n}}$
of $F$ is a distribution of compact support, $F^{\left(n\right)}(\varphi)\in\E'(\Em^{n})$.

\begin{remark}

The support of a functional $F\in\mathcal{F}(\Em)$ can be defined
by the equivalence:\[
\supp(F)\cap\supp(h)=\emptyset\quad\Leftrightarrow\quad\forall\varphi\in\mathcal{C}(\Em):\, F(\varphi+h)=F(\varphi)\,,\]
where $h\in C^{\infty}(\Em)$. Observe that if $\supp(F)\cap\supp(h)=\emptyset$
we have\[
\left\langle F^{\left(n\right)}(\varphi),h^{\otimes n}\right\rangle =\frac{d^{n}}{d\lambda^{n}}\bigg|_{\lambda=0}F(\varphi+\lambda h)=\frac{d^{n}}{d\lambda^{n}}\bigg|_{\lambda=0}F(\varphi)=0\,.\]
Conversely if $\frac{\delta F}{\delta\varphi}(h)=0$ it follows that
$F(\varphi+\lambda h)$ is invariant under (infinitesimal) changes
in $\lambda$, i.e. $F$ does not change in the {}``direction''
of $h$, $F(\varphi+h)=F(\varphi)$. It follows that\[
\supp(F)\equiv\overline{\bigcup_{\varphi\in\mathcal{C}(\Em)}\supp(F^{\left(1\right)}(\varphi))}\,,\]
where the support of the distribution $F^{\left(1\right)}(\varphi)$
is defined in the standard way (e.g. \citep[p.139]{ReedSimon1980Vol1}).
Hence the support of the $n^{\mathrm{th}}$ functional derivative
$F^{\left(n\right)}(\varphi)$ is contained in the $n$-fold Cartesian
product:\[
\supp(F^{\left(n\right)}(\varphi))\subset\supp(F)^{n}\,,\]
which is compact if $\supp(F)$ is compact.

\end{remark}

We define yet another class of so-called \emph{local functionals},
which describe local interactions

\noindent \begin{definition}\label{def:Local-Functional} A functional
of compact support, $F\in\mathcal{F}(\Em)$, is called a \emph{local
functional} if for all $n\in\mathbb{N}$

\begin{lyxlist}{MMM}
\item [{{[}LF-1]}] \noindent the support of the $n^{\mathrm{th}}$ functional
derivative of $F$ is contained in the thin diagonal $\Diag(\Em^{n}):=\left\{ \left(x_{1},\dots,x_{n}\right)\in\Em^{n}:x_{1}=\cdots=x_{n}\right\} $,\[
\supp(F^{\left(n\right)}(\varphi))\subset\Diag(\Em^{n})\,.\]

\item [{{[}LF-2]}] the wave front set%
\footnote{For the definition of the wave front set of a distribution see e.g.
\citep{Hoermander2003}. An easy example is also given in Appendix
\ref{sec:WF-delta}.%
} of $F^{\left(n\right)}(\varphi)$ is perpendicular to the tangent
bundle of the thin diagonal,\[
\WF(F^{\left(n\right)})\subset\left(T\Diag(\Em^{n})\right)^{\perp}\,.\]

\end{lyxlist}
We denote the space of local functionals by $\mathcal{F}_{\mathrm{loc}}(\Em)$.

\end{definition}

\begin{example}Typical examples for local functionals are field monomials\[
F(\varphi)=\frac{1}{k!}\int_{\Em}\left(\varphi(x)\right)^{k}f(x)\, dx\quad,\, f\in\D(\mathbb{E})\,.\]
Their functional derivatives have integral kernels of the form\[
\left(F^{\left(n\right)}(\varphi)\right)(x_{1},\dots,x_{n})=\frac{1}{\left(k-n\right)!}\left(\varphi(x_{1})\right)^{\left(k-n\right)}f(x_{1})\,\delta(x_{1}-x_{2})\cdots\delta(x_{n-1}-x_{n})\,,\]
which obviously are compactly supported on the thin diagonal, i.e.
{[LF-1]}. Furthermore their wave front set is that of Dirac's $\delta$-distribution
(see Appendix~\ref{sec:WF-delta}),\[
\WF(F^{\left(n\right)}(\varphi))=\left\{ \left(\vec{x},\vec{k}\right)\in T^{*}\Em^{n}:x_{1}=\cdots=x_{n},\,\sum_{i=1}^{n}k_{i}=0\right\} \,,\]
which is transversal to the tangent bundle of the thin diagonal%
\footnote{This for instance can be computed as the range of the differential
of the diagonal map, $\Em\rightarrow\Em^{n}:x\mapsto\left(x,\dots,x\right)$,
as done e.g. in \citep{Hoermander2003}.%
}\[
T\Diag(\Em^{n})=\left\{ \left(\vec{x},\vec{v}\right)\in T\Em^{n}:\, x_{1}=\cdots=x_{n},\, v_{1}=\cdots=v_{n}=v\right\} \,,\]
as is readily seen from the dual pairing at points ${\vec{x}\in\Diag(\Em^{n})}$.
For any $\left(\vec{x},\vec{k}\right)\in\WF(F^{\left(n\right)}(\varphi))$
and $\left(\vec{x},\vec{v}\right)\in T\Diag(\Em^{n})$ we have \[
\left\langle \vec{k},\vec{v}\right\rangle _{\vec{x}}\equiv\sum_{i=1}^{n}\left\langle k_{i},v\right\rangle =0\,,\]
hence {[LF-2]}.

\end{example}

\section{\label{sec:Partial-Algebra-of-Functionals-of-Compact-Support}The
Partial Algebra of Functionals of Compact Support}

We regard the Helmholtz operator $-\Delta+m^{2}$ on Euclidean space
$\Em$. This corresponds to the {}``Wick rotated'' Klein-Gordon
operator $\square+m^{2}$ for scalar QFT on Minkowski spacetime. The
Helmholtz operator is an elliptic partial differential operator, and
hence its fundamental solution $P$, fulfilling\begin{equation}
\left(-\Delta+m^{2}\right)P=\delta\,,\label{eq:Helmholtz-Eq}\end{equation}
in the sence of distributions, is unique up to a smooth part. This
is due to the fact that the solutions of the homogeneous equation
are smooth functions (cf. \citep[Eq. (8.1.11) and Thm. 8.3.1]{Hoermander2003}).
We choose a fixed $P$ by requiring invariance under the full Euclidean
group%
\footnote{In particular symmetry and translation invariance of $P$ are used
explicitly in the proof of Proposition \ref{pro:Partial-Algebra-of-Functionals}
and in section \ref{sub:Extension-to-the-whole-space}, respectively. %
} and Dirichlet boundary conditions at infinity, i.e. $P(x)\xrightarrow{\left\Vert x\right\Vert \rightarrow\infty}0$.
This choice is arbitrary, but corresponds to the standard one. I want
to emphasize that most of the arguments in this article do not depend
on the choice of a specific fundamental solution $P$. This applies
to all arguments depending only on the wave front set of $P$ and
in particular to the domain of definition of the Euclidean time-ordered
product, to be defined below. According to \citep[Cor. 8.3.2]{Hoermander2003}
the wave front set of any $P(x,y)$ fulfilling (\ref{eq:Helmholtz-Eq})
is that of the Dirac $\delta$-distribution,\begin{equation}
\WF(P(x,y))=\WF(\delta(x,y))=\left\{ \left(x,k_{1};x,k_{2}\right)\in T^{*}\Em^{2}:\, k_{1}+k_{2}=0\right\} \,.\label{eq:WF-P}\end{equation}

Motivated by the result for Minkowski spacetime \citep{Brunetti2009}
we define a {}``time ordering'' operator on functionals $F\in\mathcal{F}(\Em)$
by\begin{equation}
T_{\mathrm{E}}:=\exp\left(\hbar\Gamma\right)\,,\qquad\Gamma=\frac{1}{2}\int dx\, dy\, P(x,y)\,\frac{\delta^{2}}{\delta\varphi(x)\delta\varphi(y)}\,.\label{eq:T-Operator}\end{equation}

\begin{remark}There is a formal correspondence of the approach we
choose here to the more standard approach to Euclidean QFT in terms
of Gaussian functional integrals as discussed e.g. in \citep{Roepstorff1994}
and \citep{Salmhofer1999}. Namely, as the exponential of a second
order differential operator, $T_{\mathrm{E}}$ can formally be written
as the operator of convolution with a Gaussian measure with zero mean
and covariance $\hbar P$; see also the remark in the original treatment
\citep[p. 6]{Brunetti2009},\[
\left(T_{\mathrm{E}}F\right)(\varphi)=\int d\mu_{\hbar P}(\varphi-\phi)F(\phi)\,.\]
To make this more explicit, but without pondering too much about well-defined\-ness
here, we write the Gaussian measure as to be a measure in a suitably
chosen path space:\[
d\mu_{\hbar P}(\phi)=e^{\frac{1}{\hbar}S\left[\phi\right]}\mathcal{D}\phi\,,\]
where $S$ denotes the free action functional. By using the analogy
to the finite-dimensional case one computes:{\allowdisplaybreaks\begin{align*}
\int d\mu_{\hbar P}(\varphi-\phi)F(\phi) & =\int\mathcal{D}\phi\, e^{\frac{1}{\hbar}S[\varphi-\phi]}F(\phi)\\
 & =\int\mathcal{D}\phi\, e^{\frac{1}{2\hbar}\left\langle \varphi-\phi,\left(-\Delta+m^{2}\right)\varphi-\phi\right\rangle }F(\phi)\\
 & =\int\mathcal{D}\phi\,\int\mathcal{D}J\, e^{i\left\langle \varphi-\phi,J\right\rangle }\, e^{\frac{\hbar}{2}\left\langle J,\left(-\Delta+m^{2}\right)^{-1}J\right\rangle }F(\phi)\\
 & =\int\mathcal{D}J\, e^{i\left\langle \varphi,J\right\rangle }\int\mathcal{D}\phi\, e^{-i\left\langle \phi,J\right\rangle }\, e^{\frac{\hbar}{2}\left\langle J,PJ\right\rangle }F(\phi)\\
 & =\int\mathcal{D}J\, e^{i\left\langle \varphi,J\right\rangle }\int\mathcal{D}\phi\, e^{-i\left\langle \phi,J\right\rangle }\, e^{\frac{\hbar}{2}\left\langle \frac{\delta}{\delta\phi},P\frac{\delta}{\delta\phi}\right\rangle }F(\phi)\\
 & =\left[\mathfrak{F}^{-1}\mathfrak{F}\left(e^{\hbar\Gamma}F\right)\right](\varphi)=\left(T_{\mathrm{E}}F\right)(\varphi)\,,\end{align*}
}where we have written $\mathfrak{F}$ for the {}``functional Fourier
transform'' and used (\ref{eq:Helmholtz-Eq}). However, I want to
emphasize that the definition of the Euclidean time-ordering operator
$T_{\mathrm{E}}$ above is completely independent of the Gaussian
measure $d\mu_{\hbar P}$.

\end{remark}

We proceed by defining the so called Euclidean time-ordered product
for $F,G\in\mathcal{F}(\Em)$. It is obtained as a deformation of
the pointwise product $M$, $\left[M\left(F\otimes G\right)\right](\varphi):=F(\varphi)G(\varphi)$,\begin{equation}
\xymatrix{\fps{\mathcal{F}(\Em)}{\hbar}^{\otimes2}\ar[r]^{T_{\mathrm{E}}^{\otimes2}}\ar[d]_{M} & \fps{\mathcal{F}(\Em)}{\hbar}^{\otimes2}\ar@{-->}[d]^{\dE}\\
\fps{\mathcal{F}(\Em)}{\hbar}\ar[r]^{T_{\mathrm{E}}} & \fps{\mathcal{F}(\Em)}{\hbar}\,,}
\qquad F\dE G:=T_{\mathrm{E}}\circ M\circ\left(T_{\mathrm{E}}^{-1}F\otimes T_{\mathrm{E}}^{-1}G\right)\,.\label{eq:T-Prod-Operations}\end{equation}
$\mathcal{F}(\Em)$ is embedded in the space of formal power series
in $\hbar$, $\fps{\mathcal{F}(\Em)}{\hbar}$, as the component of
order $\hbar^{0}$. The time-ordering operator $T_{\mathrm{E}}$ as
well as the product $\dE$ is extended to $\fps{\mathcal{F}(\Em)}{\hbar}$
by linearity. By abuse of terminology we refer to the elements in
$\fps{\mathcal{F}(\Em)}{\hbar}$ also as functionals of compact support.

\begin{example}The inverse time ordering operator $T_{\mathrm{E}}^{-1}$
in (\ref{eq:T-Prod-Operations}) induces what is sometimes called
{}``Euclidean Wick ordering''. Take as an example the linear functionals\[
F(\varphi):=\int dx\, f(x)\,\varphi(x),\quad G(\varphi):=\int dx\, g(x)\,\varphi(x)\,,\quad f,g\in\D(\Em)\,.\]
Then\[
\left(T_{E}^{-1}(FG)\right)(\varphi)=\int dx\, dy\, f(x)\, g(y)\left[\varphi(x)\varphi(y)-\hbar P(x,y)\right]\,,\]
which can be interpreted as the Euclidean correspondence of the point
splitting approximation to normal ordering. This is why the corresponding
product (\ref{eq:T-Prod-Operations}), is often referred to as the
{}``Euclidean Wick product''. Due to its domain of definition (see
below), however, I refrain from doing so and rather call $\dE$ the
{}``Euclidean time-ordered product''.\end{example}

Observe that applying the inverse time ordering operator $T_{\mathrm{E}}^{-1}:=\exp(-\hbar\Gamma)$
before pointwise multiplication makes the tadpole terms in the expansion
of $F\dE G$ vanish (see Appendix \ref{sec:Graphs-and-Tadpoles}).
Hence we can write (\ref{eq:T-Prod-Operations}) more conveniently
as:\begin{equation}
F\dE G=\sum_{k=0}^{\infty}\frac{\hbar^{k}}{k!}\left\langle F^{\left(k\right)},P^{\otimes k}G^{\left(k\right)}\right\rangle \,,\label{eq:T-Prod-SeriesExpansion}\end{equation}
where we used the notation $P^{\otimes k}G^{\left(k\right)}$ for
the application of the map\[
\begin{array}{rccl}
P: & \E'(\Em) & \rightarrow & \D'(\Em)\\
 & f & \mapsto & P*f=\int dy\, P(\cdot,y)f(y)\end{array}\]
in each argument of $G^{\left(k\right)}$. The convolution on the
right hand side is well-defined for all $f\in\E'(\Em)$ (cf.\citep[Def. 4.2.2]{Hoermander2003}).
Hence for functionals $F\in\mathcal{F}(\Em)$ we can define:\begin{equation}
\begin{array}{rccl}
P^{\otimes k}: & \E'(\Em^{m}) & \rightarrow & \E'(\Em^{\left(m-k\right)})\otimes\D'(\Em^{k})\\
 & F^{\left(m\right)}(\varphi) & \mapsto & P^{\otimes k}F^{\left(m\right)}(\varphi)=:F_{\left(k\right)}^{\left(m-k\right)}(\varphi)\,,\end{array}\label{eq:P-map-E-D}\end{equation}
where we introduced subscript indices to denote the part of $P^{\otimes k}F^{\left(m\right)}(\varphi)$
in $\D'(\Em^{k})$. Observe that the application of $P^{\otimes k}$
does not preserve symmetry; while $F^{\left(m\right)}$ is a symmetric
distribution in $m$ variables, $F_{\left(k\right)}^{\left(m-k\right)}\equiv P^{\otimes k}F^{\left(m\right)}$
is symmetric in each set of variables separately. To be more explicit,
the part in the integral kernel representing $F_{\left(k\right)}^{\left(m-k\right)}$
is given by\[
\prod_{i=1}^{k}P(x_{i},y_{i})F^{\left(m\right)}(x_{1},\dots,x_{k};x_{k+1},\dots,x_{m})\,,\]
which is totally symmetric in $\left\{ x_{1},\dots,x_{k}\right\} $
and $\left\{ x_{k+1},\dots,x_{m}\right\} $ separately. Hence we define\begin{equation}
F_{\left(l\right)\left(m\right)}^{\left(k\right)}:=P^{\otimes l}F_{\left(m\right)}^{\left(k+l\right)}\,,\label{eq:P-notation}\end{equation}
which makes it easy to read off the (permutational) symmetry of the
distribution. Observe however that by the total symmetry of $F^{\left(k+l+m\right)}$
we have that $F_{\left(l\right)\left(m\right)}^{\left(k\right)}=F_{\left(m\right)\left(l\right)}^{\left(k\right)}$.

The product $F\dE G$ is not defined for all functionals $F,G\in\mathcal{F}(\Em)$.
This becomes apparent if we regard local functionals $F,G\in\mathcal{F}_{\mathrm{loc}}(\Em)$,
since for them the pointwise product of the distributions $F^{\left(k\right)}(\varphi)$
and $G_{\left(k\right)}(\varphi)$ is not well-defined. Explicitly
we can see this by writing the $k^{\mathrm{th}}$ term of (\ref{eq:T-Prod-SeriesExpansion})
as:\begin{equation}
\left\langle F^{\left(k\right)},P^{\otimes k}G^{\left(k\right)}\right\rangle =\int_{\Em^{2k}}dx_{1}\cdots dy_{k}\,\prod_{i=1}^{k}P(x_{i},y_{i})F^{\left(k\right)}(x_{1},\cdots,x_{k})G^{\left(k\right)}(y_{1},\cdots,y_{k})\,.\label{eq:T-Prod-Term}\end{equation}
According to the wave front set of the fundamental solution $P$ (\ref{eq:WF-P})
the product of $P(x,y)$ with itself is not defined for coinciding
points $x=y$ (cf. \citep[Thm. 8.2.10]{Hoermander2003}). There are
covectors $\left(k_{1},k_{2}\right),\left(k_{1}',k_{2}'\right)\in\left[\WF(P)\right]_{2}$
such that for $i\in\left\{ 1,2\right\} $: $k_{i}+k_{i}'=0$, hence
$\vec{0}\in\left[\WF(P)\oplus\WF(P)\right]_{2}$.%
\footnote{$\left[\WF(P)\right]_{2}$ denotes the second, i.e. covector-, component
of $\WF(P)\subset T^{*}\Em^{2}$. %
} Now, if $F$ and $G$ are local functionals, $F^{\left(k\right)}(x_{1},\dots,x_{k})$
and $G^{\left(k\right)}(y_{1},\dots,y_{k})$ have support only on
the thin diagonal $\left\{ x_{1}=\cdots=x_{k}\right\} $ and $\left\{ y_{1}=\cdots=y_{k}\right\} $,
respectively. Hence in order for the integral (\ref{eq:T-Prod-Term})
to be well-defined, we have to ask for the functionals to fulfill
\begin{equation}
\supp(F)\cap\supp(G)=\emptyset\,.\label{eq:Condition-T-Prod}\end{equation}
To sum up, we have a Euclidean time-ordered product $F\dE G$ of functionals
$F,G\in\mathcal{F}(\Em)$ which is well-defined up to the diagonal,
i.e. on $\Em^{2}\backslash\Diag(\Em^{2})$.

Observe that if (\ref{eq:Condition-T-Prod}) holds then $F\dE G$
is not a local functional. The first term in the expansion (\ref{eq:T-Prod-SeriesExpansion})
is the pointwise product $FG$ whose $n^{\mathrm{th}}$ functional
derivative is given by\[
\left(FG\right)^{\left(n\right)}=\sum_{k=0}^{n}F^{\left(k\right)}G^{\left(n-k\right)}.\]
And for non-vanishing $n$ and $k$, ${\supp(F^{\left(k\right)}G^{\left(n-k\right)})\nsubseteq\Diag(\Em^{n})}$
if $\supp(F)\cap\supp(G)=\emptyset$. A similar argument applies to
the other terms in the expansion. Nevertheless for functionals of
compact support, the support of the functional derivatives of any
term in (\ref{eq:T-Prod-SeriesExpansion}) is a Cartesian product
of compact regions and hence compact. In the above example we would
have $\supp(F^{\left(k\right)}G^{\left(n-k\right)})\subset\supp(F)^{k}\times\supp(G)^{\left(n-k\right)}$.

Regardless of the fact that the Euclidean time-ordered product is
well-defined on a subset of $\mathcal{F}(\Em)^{2}$ only, we can prove
associativity for its domain of definition.

\begin{proposition}[Partial Algebra of Functionals of Compact Support]\label{pro:Partial-Algebra-of-Functionals}
Let $F,G\in\fps{\mathcal{F}(\Em)}{\hbar}$ be functionals of compact
support. Then the Euclidean time-ordered product\[
\dE:\left(F,G\right)\mapsto F\dE G=\sum_{k=0}^{\infty}\frac{\hbar^{k}}{k!}\left\langle F^{\left(k\right)},P^{\otimes k}G^{\left(k\right)}\right\rangle \]
is well-defined in the region\[
\mathcal{D}:=\left\{ \left(F,G\right)\in\fps{\mathcal{F}(\Em)}{\hbar}^{2}:\,\supp(F)\cap\supp(G)=\emptyset\right\} ,\]
and $F\dE G\in\fps{\mathcal{F}(\Em)}{\hbar}$.

When restricted to $\mathcal{D}$, the product $\dE$ is commutative
and associative. Hence $\left(\fps{\mathcal{F}(\Em)}{\hbar},\dE\right)$
is a \emph{commutative partial algebra}, i.e. a vector space $\fps{\mathcal{F}(\Em)}{\hbar}$
with a commutative, associative product which may not be defined for
all pairs $\left(F,G\right)\in\fps{\mathcal{F}(\Em)}{\hbar}^{2}$.\end{proposition}

\begin{proof}We have already discussed the domain properties of the
product $\dE$. The commutativity follows immediately from the symmetry
of $P(x,y)$, and associativity of $\dE$ follows readily from the
associativity of the pointwise product and the definition (\ref{eq:T-Prod-Operations}),\begin{align*}
\left(F\dE G\right)\dE H & =T_{\mathrm{E}}\left\{ T_{\mathrm{E}}^{-1}\left[T_{\mathrm{E}}\left(T_{\mathrm{E}}^{-1}F\cdot T_{\mathrm{E}}^{-1}G\right)\right]\cdot T_{\mathrm{E}}^{-1}H\right\} \\
 & =T_{\mathrm{E}}\left\{ T_{\mathrm{E}}^{-1}F\cdot T_{\mathrm{E}}^{-1}G\cdot T_{\mathrm{E}}^{-1}H\right\} \\
 & =F\dE\left(G\dE H\right)\,,\end{align*}
where we assumed that ${F,G,H\in\mathcal{F}(\Em)}$ have pairwise
disjoint supports, i.e. ${\emptyset=\supp(F)\cap\supp(G)=\supp(G)\cap\supp(H)=\supp(H)\cap\supp(F)}$
such that all products in the above expressions are well-defined.
Observe that the functional equation, $e^{A}e^{B}=e^{A+B}$, holds
for the exponential $T_{\mathrm{E}}\equiv e^{\hbar\Gamma}$ due to
the symmetry of the functional derivative.\end{proof}

Observe that by writing the product $F\dE G\dE H$ in terms of its
series expansion (\ref{eq:T-Prod-SeriesExpansion}), using Cauchy's
product formula and the Leibniz rule, the graph structure of the expansion
becomes immediately apparent:\[
F\dE G\dE H=\sum_{n=0}^{\infty}\frac{\hbar^{n}}{n!}\sum_{m=k}^{n}\sum_{k=0}^{m}{n \choose m}{m \choose k}\left\langle F^{\left(k+m-k\right)}G_{\left(k\right)}^{\left(n-m\right)}H_{\left(n-m\right)\left(m-k\right)}\right\rangle \,.\]
The terms on the right hand side correspond to graphs with three vertices,
$F$, $G$, $H$, where $k$ lines connect $F$ and $G$, $\left(m-k\right)$
edges connect $F$ and $H$ and there are $\left(n-m\right)$ lines
between $G$ and $H$, see also Appendix \ref{sec:Graphs-and-Tadpoles}.

\section{\label{sec:Renormalization}Renormalization}

Associativity makes it possible to speak of $n$-fold time-ordered
products\begin{equation}
E_{n}(F_{1}\otimes\cdots\otimes F_{n}):=F_{1}\dE\cdots\dE F_{n}\,,\label{eq:En}\end{equation}
These are linear maps\[
E_{n}:\fps{\mathcal{F}(\Em)}{\hbar}^{\otimes n}\rightarrow\fps{\mathcal{F}(\Em)}{\hbar}\,,\]
which are well-defined, if the supports of the functionals $F_{1},\dots,F_{n}\in\mathcal{F}(\Em)$
are pairwise disjoint, i.e. \begin{equation}
\supp(F_{i})\cap\supp(F_{j})=\emptyset\quad\forall i,j\in\left\{ 1,\dots,n\right\} ,\, i\neq j\,.\label{eq:Condition-En-welldef}\end{equation}
In order to be able to properly define the coefficients $S_{\mathrm{E}}^{\left(n\right)}(V^{\otimes n})\equiv E_{n}(V^{\otimes n})$
in the expansion of the Euclidean $S$-matrix (cf. (\ref{eq:S-Matrix})),
we have to extend the maps $E_{n}$ towards functionals with arbitrary
support properties. In the presented formalism this is possible for
local functionals only. The extension is performed by applying the
recursive procedure of Epstein and Glaser. In each recursion step
Epstein and Glaser use the \emph{causality condition} to define the
time-ordered products up to the thin diagonal, \emph{translation invariance}
to define $E_{n}$ for all points except the origin and in the last
step include the origin in the domain of a newly defined time-ordered
product. It is this last step, which corresponds to renormalization.
The freedom in the definition of the new time-ordered product is governed
by the theory of \emph{extension of distributions}.

As already described in the introduction, in the Euclidean framework
we have to find a suitable replacement for the causality condition,
in order to make the Epstein-Glaser recursion applicable.

We first define\begin{equation}
\forall F\in\mathcal{F}(\Em):\quad E_{0}(F)=\openone\quad\mbox{and}\quad E_{1}(F)=F\,.\label{eq:Induction-basis}\end{equation}
This serves as the induction basis and already implies that $E_{2}:\fps{\mathcal{F}_{\loc}(\Em)}{\hbar}^{\otimes2}\rightarrow\fps{\mathcal{F}(\Em)}{\hbar}$
is symmetric and uniquely defined up to the diagonal $\Diag(\Em^{2})$.
Assuming that $E_{k}$ is properly defined for all $k<n$ on the whole
of $\Em^{k}$, makes it possible, using a certain factorization property
(see below), to uniquely define the $n$-fold product $E_{n}$ on
$\Em^{n}\backslash\Diag(\Em^{n})$. The last step, which makes the
whole argument valid, is to show that $E_{n}$ can be extended to
the whole space $\Em^{n}$.

\subsection{Construction up to the thin diagonal}

For the recursive construction of Epstein and Glaser - as well as
for its generalizations - the causality condition for the time-ordered
product is crucial. Since we cannot make use of this condition in
a Euclidean framework, we have to replace it by another one, which
we call \emph{Euclidean causality}, and which makes also sense on
general Riemannian manifolds.

\begin{condition}[Euclidean Causality]\label{con:Eucl-Causality}
Let $I\subset\left\{ 1,\dots,k\right\} $ be a subset of the index
set $\left\{ 1,\dots,k\right\} $ with non-empty complement $I^{c}$.
If for all $i\in I$ and for all $j\in I^{c}$ the supports of the
corresponding functionals are disjoint, \[
\forall i\in I,\,\forall j\in I^{c}:\quad\supp(F_{i})\cap\supp(F_{j})=\emptyset\,,\]
then the $k$-fold Euclidean time-ordered product has the following
factorization property: \[
E_{k}(F_{1}\otimes\cdots\otimes F_{k})=E_{\left|I\right|}(\bigotimes_{i\in I}F_{i})\dE E_{\left|I^{c}\right|}(\bigotimes_{j\in I^{c}}F_{j})\,.\]
\end{condition}

Having this supplement for causality, we can start the induction procedure.

\begin{IndHyp}

\noindent We assume that for all $k<n$ the maps $E_{k}$ are

\begin{itemize}
\item properly defined on the whole of $\fps{\mathcal{F}_{\loc}(\Em)}{\hbar}^{\otimes k}$,
\item symmetric:\[
\forall\pi\in\mathbb{S}(k):\quad E_{k}(F_{\pi(1)}\otimes\cdots\otimes F_{\pi(k)})=E_{k}(F_{1}\otimes\cdots\otimes F_{k})\,,\]

\item and fulfill Euclidean causality, i.e. Condition \ref{con:Eucl-Causality},
for all $k<n$.
\end{itemize}
\end{IndHyp}

This already determines the $n^{\mathrm{th}}$ order maps $E_{n}$
uniquely up to the thin diagonal:

\begin{proposition}\label{pro:Uniqueness-up-to-Diagonal}Let $E_{k}:\fps{\mathcal{F}_{\loc}(\Em)}{\hbar}^{\otimes k}\rightarrow\fps{\mathcal{F}(\Em)}{\hbar}$
fulfill the induction hypothesis for all $k<n$. Then the $n^{\mathrm{th}}$
order map \[
\begin{array}{rccl}
E_{n}: & \fps{\mathcal{F}_{\loc}(\Em)}{\hbar}^{\otimes n} & \rightarrow & \fps{\mathcal{F}(\Em)}{\hbar}\\
 & \sum_{i}F_{1}^{i}\otimes\cdots\otimes F_{n}^{i} & \mapsto & \sum_{i}E_{n}(F_{1}^{i}\otimes\cdots\otimes F_{n}^{i})\,.\end{array}\]
is uniquely determined for all functional tensors, $\sum_{i}F_{1}^{i}\otimes\cdots\otimes F_{n}^{i}$,
with \[
\bigcup_{i}\supp(F_{1}^{i}\otimes\cdots\otimes F_{n}^{i})\cap\Diag(\Em^{n})=\emptyset\,.\]
\end{proposition}

\begin{proof}Condition \ref{con:Eucl-Causality} makes it possible
to follow closely the proof of \citep{Brunetti2000}. Let $\mathcal{I}=\left\{ I\subsetneq\left\{ 1,\dots,n\right\} \right\} $
and define neighborhoods\begin{equation}
U_{I}:=\left\{ \left(x_{1},\dots,x_{n}\right)\subset\Em^{n}\backslash\Diag(\Em^{n}):\, x_{i}\neq x_{j}\,\forall i\in I,\forall j\in I^{c}\right\} \,.\label{eq:Neighborhoods-UI}\end{equation}
Then $\left\{ U_{I}:\, I\in\mathcal{I}\right\} $ is a cover for $\Em^{n}\backslash\Diag(\Em^{n})$,
that is\begin{equation}
\bigcup_{I\in\mathcal{I}}U_{I}=\Em^{n}\backslash\Diag(\Em^{n})\,.\label{eq:Cover-UI}\end{equation}
The inclusion $\bigcup_{I\in\mathcal{I}}U_{I}\subset\Em^{n}\backslash\Diag(\Em^{n})$
is obvious. To show the inclusion in the other direction let $\left(x_{1},\dots,x_{n}\right)\in\Em^{n}\backslash\Diag(\Em^{n})$.
Then for at least one pair $\left(i,j\right)$ we have that $x_{i}\neq x_{j}$.
Defining $I=\left\{ k\in\left\{ 1,\dots n\right\} :\, x_{k}=x_{i}\right\} $,
we have $\left(x_{1},\dots,x_{n}\right)\in U_{I}$, hence the inclusion
in the opposite direction.

Now that we dispose of the cover $\left\{ U_{I},\, I\in\mathcal{I}\right\} $,
observe the equivalence of the assertions\[
\forall i\in I,\forall j\in I^{c}:\quad\supp(F_{i})\cap\supp(F_{j})=\emptyset\]
and\[
\supp(F_{1}\otimes\cdots\otimes F_{n})\subset U_{I}\,.\]
By using the induction hypothesis, we are able to define $n$-fold
time-ordered products on $U_{I}$, for all ${F_{1}\otimes\cdots\otimes F_{n}\in\fps{\mathcal{F}(\Em)}{\hbar}^{\otimes n}}$
with ${\supp(F_{1}\otimes\cdots\otimes F_{n})\subset U_{I}}$ we set:
\begin{align}
E_{n}^{I}(F_{1}\otimes\cdots\otimes F_{n}):=E_{\left|I\right|}(\bigotimes_{i\in I}F_{i})\dE E_{\left|I^{c}\right|}(\bigotimes_{j\in I^{c}}F_{j})\,.\label{eq:Def-EnI-on-UI}\end{align}
Where the right hand side is well-defined since the maps $E_{\left|I\right|}$
for $\left|I\right|<n$ have already been defined by assumption and
${\supp(F_{1}\otimes\cdots\otimes F_{n})\subset U_{I}}$ implies ${\supp(E_{\left|I\right|}(\bigotimes_{i\in I}F_{i}))\cap\supp(E_{\left|I^{c}\right|}(\bigotimes_{j\in I^{c}}F_{j}))=\emptyset}$.

We have to make sure that on the overlaps $U_{I}\cap U_{J}$, the
maps $E_{n}^{I}$ and $E_{n}^{J}$ coincide:\begin{equation}
E_{n}^{I}\Big|_{U_{I}\cap U_{J}}=E_{n}^{J}\Big|_{U_{I}\cap U_{J}}\,.\label{eq:EnI-Sheaf-consistency}\end{equation}
Again from the induction hypothesis it follows that for all ${F_{1}\otimes\cdots\otimes F_{n}}$
with ${\supp(F_{1}\otimes\cdots\otimes F_{n})\subset U_{I}\cap U_{J}}$
we have:\[
E_{\left|I\right|}(\bigotimes_{i\in I}F_{i})=E_{\left|I\cap J\right|}(\bigotimes_{k\in I\cap J}F_{k})\dE E_{\left|I\cap J^{c}\right|}(\bigotimes_{l\in I\cap J^{c}}F_{l})\]
and analogously for $E_{\left|J\right|}$. Hence we have{\allowdisplaybreaks\begin{align*}
E_{n}^{I}(F_{1}\otimes\cdots\otimes F_{n}) & =E_{\left|I\right|}(\bigotimes_{i\in I}F_{i})\dE E_{\left|I^{c}\right|}(\bigotimes_{j\in I^{c}}F_{j})\\
 & =E_{\left|I\cap J\right|}\dE E_{\left|I\cap J^{c}\right|}\dE E_{\left|I^{c}\cap J\right|}\dE E_{\left|I^{c}\cap J^{c}\right|}\\
 & =E_{\left|I\cap J\right|}\dE E_{\left|I^{c}\cap J\right|}\dE E_{\left|I\cap J^{c}\right|}\dE E_{\left|I^{c}\cap J^{c}\right|}\\
 & =E_{\left|J\right|}(\bigotimes_{i\in I}F_{i})\dE E_{\left|J^{c}\right|}(\bigotimes_{j\in J^{c}}F_{j})\\
 & =E_{n}^{J}(F_{1}\otimes\cdots\otimes F_{n})\,,\end{align*}
}where we used the symmetry of $\dE$ and have omitted the arguments
in the second and third row. Observe that if $I\cap J=\emptyset$,
the argument is still valid by (\ref{eq:Induction-basis}).

The individual time-ordered products $E_{n}^{I}$ defined on the sets
of the open cover $\left\{ U_{I}:\, I\in\mathcal{I}\right\} $ now
need to be {}``glued together'' to give \emph{one} time-ordered
product $E_{n}^{0}$ on $\Em^{n}\backslash\Diag(\Em^{n})$. A standard
way to achieve this, in the case when the time-ordered products are
distributions, is to introduce a partition of unity $\left\{ \chi_{I}:\, I\in\mathcal{I}\right\} $
subordinate to $\left\{ U_{I}:\, I\in\mathcal{I}\right\} $, and to
define the \emph{unique} time-ordered product on $\Em^{n}\backslash\Diag(\Em^{n})$
as the as the weighted sum of the individual time-ordered products
on $U_{I}$ weighted with $\chi_{I}$; see \citep[Sec. 4]{Brunetti2000}
for details. Observe, however, that in contrast to \citep{Brunetti2000}
the time-ordered products $\left\{ E_{n}^{I}(F_{1}\otimes\cdots\otimes F_{n}):\, I\in\mathcal{I}\right\} $
we are dealing with here are functionals on $\mathcal{C}(\Em)$ rather
than distributions on $\Em^{n}\backslash\Diag(\Em^{n})$. In particular,
there is no ad~hoc notion of a product of the functional $E_{n}^{I}(F_{1}\otimes\cdots\otimes F_{n})$
with a smooth function, $\chi_{I}$ say. Hence for the gluing of $\left\{ E_{n}^{I},\, I\in\mathcal{I}\right\} $
we cannot use the standard method of \citep{Brunetti2000}. Instead
we implement an argument given in \citep{Brunetti2009}, which, as
well as the original reasoning, is conclusive only for local functionals.
Let $\sum_{i}F_{1}^{i}\otimes\cdots\otimes F_{n}^{i}\in\fps{\mathcal{F}_{\loc}(\Em)}{\hbar}^{\otimes n}$
such that\begin{equation}
\forall i:\,\supp(F_{1}^{i}\otimes\cdots\otimes F_{n}^{i})\cap\Diag(\Em^{n})=\emptyset\,,\label{eq:LocFuncSuppOutsideDiag}\end{equation}
by abuse of notation we write $\sum F_{1}\otimes\cdots\otimes F_{n}\in\mathcal{F}_{\loc}(\Em)^{\otimes n}\backslash\Diag(\Em^{n})$
in this case. We now want to define the product $E_{n}$ for those
elements of $\mathcal{F}_{\loc}(\Em)^{\otimes n}\backslash\Diag(\Em^{n})$
whose support is not contained in any neighborhood of the cover $\left\{ U_{I},\, I\in\mathcal{I}\right\} $,
$\forall I\in\mathcal{I}$: $\supp(\sum F_{1}\otimes\cdots\otimes F_{n})\nsubseteq U_{I}$.
The crucial fact, to be shown below, is that any element in $\mathcal{F}_{\loc}(\Em)^{\otimes n}\backslash\Diag(\Em^{n})$
can be written as a finite sum of tensor products of local functionals,
which are fully supported inside some neighborhood $U_{I}$. For these
tensor products the map $E_{n}^{I}$ is already defined by (\ref{eq:Def-EnI-on-UI}).
It is unique due to the sheaf property (\ref{eq:EnI-Sheaf-consistency}).
This definition is then extended to the sum by linearity. So what
remains to be shown is the decomposition property for $\mathcal{F}_{\loc}(\Em)^{\otimes n}\backslash\Diag(\Em^{n})$.
Let $\sum_{i}F_{1}^{i}\otimes\cdots\otimes F_{n}^{i}\in\mathcal{F}_{\loc}(\Em)^{\otimes n}\backslash\Diag(\Em^{n})$,
then by (\ref{eq:LocFuncSuppOutsideDiag}) we have that\begin{equation}
\forall i:\;\bigcap_{k=1}^{n}\supp(F_{k}^{i})=\emptyset\,.\label{eq:SupportsDoNotIntersect}\end{equation}
The $F_{k}^{i}$ are local functionals, and hence can be written as
a finite sum of local functionals of arbitrarily small support (cf.
\citep[Lem. 3.2]{Brunetti2009}),%
\footnote{Although the definition of a local functional in \citep{Brunetti2009}
differs from the one given in this article, it can be shown that both
definitions are equivalent, see Appendix \ref{app:Local-Functionals}
and also \citep{BrunettiFredenhagenRibeiro2009}. Hence the results
of \citep{Brunetti2009} on local functionals, and Lemma 3.2 in particular,
are applicable in our context.%
}\begin{equation}
F_{k}^{i}=\sum_{r_{i,k}}s^{r_{i,k}}F_{k}^{i,r_{i,k}}\,,\quad s^{r_{i,k}}\in\left\{ +,-\right\} \,.\label{eq:LocFuncExpansion}\end{equation}
Because of (\ref{eq:SupportsDoNotIntersect}) the supports $\supp(F_{k}^{i,r_{i,k}})$
can be chosen in such a way that at most $\left(n-1\right)$ of them
intersect. To be more precise this means that for each pair $\left(i,r_{i}\right)$,
$r_{i}\in\mathbb{N}^{n}$ there is some index set $I_{\left(i,r_{i}\right)}\in\mathcal{I}$
such that\begin{equation}
\supp(F_{1}^{i,r_{i,1}}\otimes\cdots\otimes F_{n}^{i,r_{i,n}})\subset U_{I_{\left(i,r_{i}\right)}}\,,\label{eq:SupportInUI}\end{equation}
see Figure \ref{fig:Deviding-Supports-of-Local-Functionals}. Since
the $n$-fold Euclidean time-ordered product is uniquely defined on
these neighborhoods, we can define for any element of $\mathcal{F}_{\loc}(\Em^{n})^{\otimes n}\backslash\Diag(\Em^{n})$:\[
E_{n}^{0}(\sum_{i}F_{1}^{i}\otimes\cdots\otimes F_{n}^{i}):=\sum_{i,r_{i}}\left(\prod_{k=1}^{n}s^{r_{i,k}}\right)E_{n}^{I_{\left(i,r_{i}\right)}}(F_{1}^{i,r_{i,1}}\otimes\cdots\otimes F_{n}^{i,r_{i,n}})\,.\]

Thus we have reached a definition of the Euclidean time-ordered product
up to the thin diagonal. We introduce the notation $E_{n}^{0}$ for
this product in order to distinguish it from its extension to the
whole space, we aim at constructing. For the first part of the induction,
i.e. Proposition \ref{pro:Uniqueness-up-to-Diagonal}, it remains
to be shown that the definition is independent of the choice of the
expansion (\ref{eq:LocFuncExpansion}), that the maps $E_{n}^{0}$
are symmetric and that they fulfill Euclidean causality (Condition
\ref{con:Eucl-Causality}) for $k=n$.

\emph{Independence~of~expansion~(\ref{eq:LocFuncExpansion}).}
Taking another expansion, also fulfilling (\ref{eq:SupportInUI}),
corresponds to taking different index sets $I_{\left(i,r_{i}\right)}$,
i.e. different neighborhoods, for the definition of $E_{n}^{0}$.
However the Euclidean time-ordered product is uniquely defined on
the intersections of these neighborhoods due to the sheaf property
(\ref{eq:EnI-Sheaf-consistency}).

\emph{Symmetry.} For the definition of the maps $E_{n}^{0}$ for permuted
arguments one can take the maps $E_{n}^{I_{\pi(r)}}$ defined on the
neighborhoods $U_{I_{\pi(r)}}$, \begin{align*}
E_{n}^{0}(F_{\pi(1)}\otimes\cdots\otimes F_{\pi(n)}) & =\sum_{r=\left(r_{k}\right)}\left(\prod_{k=1}^{n}s^{r_{k}}\right)E_{n}^{I_{\pi(r)}}(F_{\pi(1)}^{r_{\pi(1)}}\otimes\cdots\otimes F_{\pi(n)}^{r_{\pi(n)}})\\
 & =\sum_{\pi(r)=\left(r_{\pi(k)}\right)}\left(\prod_{k=1}^{n}s^{r_{k}}\right)E_{n}^{I_{r}}(F_{1}^{r_{1}}\otimes\cdots\otimes F_{n}^{r_{n}})\\
 & =E_{n}^{0}(F_{1}\otimes\cdots\otimes F_{n})\,.\end{align*}

\emph{Euclidean~causality.} If for all $i\in I$ and for all $j\in I^{c}$
we have $\supp(F_{i})\cap\supp(F_{j})=\emptyset$, then $\supp(F_{1}\otimes\cdots\otimes F_{n})\subset U_{I}$
and\[
E_{n}^{0}(F_{1}\otimes\cdots\otimes F_{n})=E_{n}^{I}(F_{1}\otimes\cdots\otimes F_{n})=E_{\left|I\right|}(\bigotimes_{i\in I}F_{i})\dE E_{\left|I^{c}\right|}(\bigotimes_{j\in I^{c}}F_{j})\,.\]

\end{proof}

\begin{figure}[h]
\includegraphics[width=5cm,keepaspectratio]{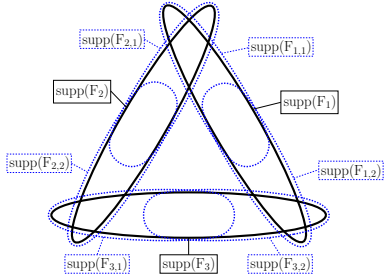}

\caption{\label{fig:Deviding-Supports-of-Local-Functionals}Dividing the supports
of local functionals helps with defining $E_{n}^{0}$. In the picture
$\supp(F_{1}\otimes F_{2}\otimes F_{3})\cap\Diag(\Em^{3})=\emptyset$
but there is no index set $I\subset\left\{ 1,2,3\right\} $ such that
$\supp(F_{1}\otimes F_{2}\otimes F_{3})\subset U_{I}$. However, for
any three of the functionals of smaller support such a neighborhood
can always be found, $\forall\left(i,j,k\right)$: $\supp(F_{1,i}\otimes F_{2,j}\otimes F_{3,k})\subset U_{I_{\left(i,j,k\right)}}$.}
\end{figure}

Up to here we have constructed the $n$-fold Euclidean time-ordered
products ${E_{n}^{0}(\sum F_{1}\otimes\cdots\otimes F_{n})}$ up to
the thin diagonal, under the assumption that the maps $E_{k}$ where
already defined on the whole space $\Em^{k}$ for all $k<n$. So what
remains to be done, is to prove that for each $n\in\mathbb{N}$ the
map $E_{n}^{0}$ can be extended to the whole space $\Em^{n}$.

\subsection{\label{sub:Extension-to-the-whole-space}Extension to the whole space}

In their original work \citep{Epstein1973}, Epstein and Glaser reduce
the problem of extending the algebra-valued distributions  $T(\L(x_{1})\dots\L(x_{q}))$
(adapting to their notation here) to an extension problem for scalar
distributions by expanding the $n$-fold time-ordered product in terms
of Wick products of the fields (cf. formulas (42)/(43) loc. cit.).
We will show in this section that, using the tools of \citep{Brunetti2009},
an analogous expansion can be done in the Euclidean framework. The
fact that we regard local functionals is essential.

\subsubsection{Expansion Formula of Epstein-Glaser in (Euclidean) pAQFT}

The definition of local functionals (Def. \ref{def:Local-Functional})
implies that the integral kernel of the $n^{\mathrm{th}}$ functional
derivative of any $F\in\mathcal{F}_{\loc}(\mathbb{E})$ can be written
as (cf. \citep{Brunetti2000})\begin{equation}
F^{\left(n\right)}(\varphi)(x_{1},\dots,x_{n})=\sum_{k}f_{\varphi}^{n,k}(x)\, p_{k}(\partial_{\vec{r}})\,\delta(\vec{r})\,,\quad f_{\varphi}^{n,k}\in\D(\Diag(\Em^{n}))\cong\D(\Em)\label{eq:local-functional-CM-Rel-Coord}\end{equation}
where $x=\frac{1}{n}\sum_{i=1}^{n}x_{i}$ is the {}``center of mass''-coordinate,
$\vec{r}=\left(r_{1},\dots,r_{n-1}\right)$ are relative coordinates
and $\left(p_{k}\right)_{k\in\mathbb{N}}$ is a basis of homogeneous,
symmetric polynomials in $\left(n-1\right)$ variables.%
\footnote{We do not specify the relative coordinates any further, but will assume
them to be chosen in such a way that the product measure on $\Em^{n}$
is invariant,\[
dx_{1}(x,\vec{r})\cdots dx_{n}(x,\vec{r})=dx\, dr_{1}\cdots dr_{n-1}\equiv dx\, d\vec{r}\,,\]
which is always possible. For an explicit choice see \citep[Prop. 3.1]{DuetschFredenhagen2004}. %
} Equation (\ref{eq:local-functional-CM-Rel-Coord}) is equivalent
to saying that the functional derivatives, $F^{\left(n\right)}(\varphi)$,
can be restricted to surfaces which are transversal to the thin diagonal
$\Diag(\Em^{n})$, which is implied by the second condition {[LF-2]}
of Def. \ref{def:Local-Functional} (cf. \citep[Cor. 8.2.7]{Hoermander2003},
\citep[Lem. 6.1]{Brunetti2000}). By using (\ref{eq:local-functional-CM-Rel-Coord})
we find a compact formula for the Taylor expansion up to order $N$
of a local functional at some reference field configuration $\varphi_{0}$:\begin{align}
F_{\varphi_{0}}^{\left[N\right]}(\varphi) & =\sum_{n=0}^{N}\frac{1}{n!}\left\langle F^{\left(n\right)}(\varphi_{0}),\left(\varphi-\varphi_{0}\right)^{\otimes n}\right\rangle \nonumber \\
 & =\sum_{n=0}^{N}\sum_{k}\left\langle f_{\varphi_{0}}^{n,k}(x),\left\langle \delta(\vec{r}),p_{k}(-\partial_{\vec{r}})\frac{\left(\varphi-\varphi_{0}\right)^{\otimes n}(x,\vec{r})}{n!}\right\rangle \right\rangle \nonumber \\
 & =\sum_{n=0}^{N}\sum_{k}\left\langle f_{\varphi_{0}}^{n,k}(x),A_{\varphi-\varphi_{0}}^{n,k}(x)\right\rangle \,,\label{eq:LocalFunctional-TayApprox}\end{align}
where we introduced the so-called \emph{balanced fields} (cf. \citep{Buchholz2002})\[
A_{\varphi-\varphi_{0}}^{n,k}(x)=\left\langle \delta(\vec{r}),p_{k}(-\partial_{\vec{r}})\frac{\left(\varphi-\varphi_{0}\right)^{\otimes n}(x,\vec{r})}{n!}\right\rangle \,.\]

Following \citep{Brunetti2009} we impose two further conditions on
the Euclidean $S$-matrix \begin{equation}
S_{\mathrm{E}}(F):=\exp_{\dE}(F)\equiv\sum_{n=0}^{\infty}\frac{1}{n!}E_{n}(F^{\otimes n})\,.\label{eq:SMatrixEucl}\end{equation}
Let $F\in\mathcal{F}_{\loc}(\Em)$. The first condition, called \emph{$\varphi$-locality},
states that the $S$-matrix $S_{\mathrm{E}}(F)(\varphi_{0})$ at a
given field configuration $\varphi_{0}$ should depend only on the
Taylor expansion of $F$ around $\varphi_{0}$.

\begin{condition}[$\varphi$-locality]\label{con:Phi-Locality} $S_{\mathrm{E}}(F)(\varphi_{0})=S_{\mathrm{E}}(F_{\varphi_{0}}^{\left[N\right]})(\varphi_{0})+\mathcal{O}(\hbar^{N+1})$.\end{condition}

This condition makes it possible to treat not only polynomial, but
also more general functions of the fields as interaction functionals
in pAQFT. Since the $S$-matrix in perturbation theory is defined
in terms of a formal power series in $\hbar$, according to Condition~\ref{con:Phi-Locality}
$S_{\mathrm{E}}(F)(\varphi_{0})$ is (up to the renormalization freedom)
fully determined by the Taylor expansion of $F$ around $\varphi_{0}$,
because the additional terms are required to be of sufficiently high
order in $\hbar$. Non-polynomial interactions where excluded in the
original treatment by Epstein and Glaser. However, their consistent
incorporation in the pertubative treatment of QFT is desireable not
only from the viewpoint of non-polynomial models like, for instance,
the Sine-Gordon model. They also seem to be necessary in a perturbative
treatment of super-symmetric extensions of the standard model.%
\footnote{Private communication with K. Fredenhagen. See also \citep{GrigoreScharf2008,Sibold2008}.%
} Observe, however, that Condition~\ref{con:Phi-Locality} is a condition
within perturbation theory, which aims at a consistent treatment of
the topic, rather than an extension to the non-perturbative regime.

The second condition, \emph{field independence}, states that $S$
should only depend implicitly (i.e. via the interaction) on the field
configuration $\varphi$, which makes the chain rule easily applicable:

\begin{condition}[Field independence]\label{con:Field-Independence}
$\forall g\in\E(\Em)$: $\left\langle \frac{\delta S_{\mathrm{E}}(F)}{\delta\varphi},g\right\rangle =S_{\mathrm{E}}^{\left(1\right)}(F)\left\langle \frac{\delta F}{\delta\varphi},g \right\rangle $.\end{condition}

Using Conditions \ref{con:Phi-Locality} and \ref{con:Field-Independence}
one can consistently insert the approximation (\ref{eq:LocalFunctional-TayApprox})
for $F$ into (\ref{eq:SMatrixEucl}). This, as shown below, reduces
the extension problem of the functionals $E_{n}^{0}$ to that of certain
scalar coefficients in the {}``Wick expansion'' of $E_{n}(F^{\otimes n})$,
\begin{align*}
E_{n}(F^{\otimes n})(\varphi) & =E_{n}(F_{\varphi_{0}}^{\left[N\right]}\otimes\cdots\otimes F_{\varphi_{0}}^{\left[N\right]})(\varphi)+\mathcal{O}\left(\hbar^{n\left(N+1\right)}\right)\\
 & =\sum_{l_{j},k_{j}}E_{n}(\left\langle f_{\varphi_{0}}^{l_{1},k_{1}}(x_{1}),A_{\varphi-\varphi_{0}}^{l_{1},k_{1}}(x_{1})\right\rangle \otimes\cdots\otimes\left\langle f_{\varphi_{0}}^{l_{n},k_{n}}(x_{n}),A_{\varphi-\varphi_{0}}^{l_{n},k_{n}}(x_{n})\right\rangle )\\
 & =\sum_{l_{j},k_{j}}\int dx_{1}\cdots dx_{n}\,\prod_{j=0}^{n}f_{\varphi_{0}}^{l_{j},k_{j}}(x_{j})\, E_{n}(A_{\varphi-\varphi_{0}}^{l_{1},k_{1}}(x_{1})\otimes\cdots\otimes A_{\varphi-\varphi_{0}}^{l_{n},k_{n}}(x_{n}))\,.\end{align*}
Setting $\varphi=\varphi_{0}$ and defining the scalar distributions\begin{equation}
t^{l,k}(x_{1},\dots,x_{n}):=E_{n}(A_{\varphi-\varphi_{0}}^{l_{1},k_{1}}(x_{1})\otimes\cdots\otimes A_{\varphi-\varphi_{0}}^{l_{n},k_{n}}(x_{n}))\bigg|_{\varphi=\varphi_{0}}\,,\quad l,k\in\mathbb{N}^{n}\,,\label{eq:ScalarTProd}\end{equation}
we arrive at\begin{equation}
E_{n}(F^{\otimes n})(\varphi_{0})=\sum_{l_{j},k_{j}}\int dx_{1}\cdots dx_{n}\,\prod_{j=0}^{n}f_{\varphi_{0}}^{l_{j},k_{j}}(x_{j})\, t^{l,k}(x_{1},\dots,x_{n})\,.\label{eq:Wick-expansion}\end{equation}
Hence the problem of defining the coefficients of the $S$-matrix,
$E_{n}(F^{\otimes n})(\varphi_{0})$, for local interactions $F\in\mathcal{F}_{\loc}(\mathbb{E})$
is reduced to extending the coefficients $t^{l,k}$ (\ref{eq:ScalarTProd}).%
\footnote{In the case of polynomial interactions in MQFT formula (\ref{eq:Wick-expansion})
reduces to the familiar Wick expansion formula of Epstein and Glaser
(cf. \citep[Ex. on p. 19]{Brunetti2009}).%
} These are scalar distributions, which in general are well-defined
up to the thin diagonal. Observe however, that they only depend on
differences of the variables $\left\{ x_{1},\dots,x_{n}\right\} $,
since the fundamental solution $P(x-y)$ with respect to which the
Euclidean time-ordered products $E_{n}$ were defined depend only
on relative distances. Hence in flat Euclidean space the distributions
$t^{l,k}$ are translation invariant (along diagonal directions),
i.e. under transformations $\left(x_{1},\dots,x_{n}\right)\mapsto\left(x_{1}+a,\dots,x_{n}+a\right)$.
Consequently, for the extension to the thin diagonal it suffices to
define them at the origin.

\begin{example}To illustrate this procedure in a graphical setting
let us regard the following example\[
\begin{array}{c}
\\\mbox{\FthreeGoneHtwoF}\\
\\\end{array}\qquad\qquad\qquad\begin{array}{rl}
F(\varphi) & =\frac{1}{5!}\int dx\, f(x)\left(\varphi(x)\right)^{5}\,,\\
G(\varphi) & =\frac{1}{4!}\int dx\, g(x)\left(\varphi(x)\right)^{4}\,,\\
H(\varphi) & =\frac{1}{3!}\int dx\, h(x)\left(\varphi(x)\right)^{3}\,,\end{array}\]
with $f,g,h\in\D(\Em)$ test functions of compact support. The integral
kernels of the functional derivatives are given by:\begin{align*}
\left(F^{\left(5\right)}(\varphi)\right)(x_{1},\dots,x_{5}) & =f(x_{1})\,\delta(x_{1}-x_{2})\cdots\delta(x_{4}-x_{5})\\
\left(G^{\left(4\right)}(\varphi)\right)(y_{1},\dots,y_{4}) & =g(y_{1})\,\delta(y_{1}-y_{2})\cdots\delta(y_{3}-y_{4})\\
\left(H^{\left(3\right)}(\varphi)\right)(z_{1},\dots,z_{3}) & =h(z_{1})\,\delta(z_{1}-z_{2})\,\delta(z_{2}-z_{3})\,.\end{align*}
Hence the corresponding amplitude for the above graph is given by\[
\int dx\, dy\, dz\,\left(P(x-y)\right)^{3}\left(P(x-z)\right)^{2}P(y-z)\, f(x)\, g(y)\, h(z)\,,\]
where in the first induction step the translation invariant distributions
$\left(P(x-y)\right)^{3}$ and $\left(P(x-z)\right)^{2}$ have to
be extended to the origin, giving renormalized distributions $\left(P(x-y)\right)_{\mathrm{ren}}^{3}$
and $\left(P(x-z)\right)_{\mathrm{ren}}^{2}$, respectively. In the
second step the domain of the equally translation invariant distribution
\[
t_{0}(x,y,z)=\left(P(x-y)\right)_{\mathrm{ren}}^{3}\left(P(x-z)\right)_{\mathrm{ren}}^{2}P(y-z)\]
 has to be extended to give the renormalized amplitude $t_{\mathrm{ren}}$.\end{example}

\subsubsection{Extension to the origin}

The extension problem for scalar distributions, is well understood
and can conveniently be discussed in terms of the scaling degree \citep{Steinmann1971}
of the corresponding distribution.

\begin{definition}[scaling degree; cf. \cite{Brunetti2000}] Let $\dim(\Em)=d$
and define \[
\begin{array}{rccl}
\Lambda: & \mathbb{R}_{+}\times\D(\Em) & \rightarrow & \D(\Em)\\
 & \left(\lambda,\phi\right) & \mapsto & \phi^{\lambda}:=\lambda^{-d}\phi(\lambda^{-1}\cdot)\end{array}\]
to be the action of the positive real numbers on test functions. This
induces, via the pullback, the action on distributions. For $t\in\D'(\Em)$
we have:\[
t_{\lambda}(\phi):=\left(\Lambda^{*}t\right)(\phi)=t(\phi^{\lambda})\,.\]
The scaling degree $\sd(t)$ of a distribution $t$ (with respect
to the origin) is defined to be\[
\sd(t):=\inf\left\{ \omega'\in\mathbb{R}:\,\lim_{\lambda\searrow0}\lambda^{\omega'}t_{\lambda}=0\in\D'(\Em)\right\} \,.\]
\end{definition}

The extension of the scalar distributions is governed by the following
theorem due to Brunetti and Fredenhagen

\begin{theorem}[cf. \cite{Brunetti2000}] Let $t_{0}\in\D'(\mathbb{R}^{d}\backslash\left\{ 0\right\} )$
have scaling degree $\sd(t_{0})$ with respect to the origin. Let 

\begin{itemize}
\item $\sd(t_{0})<d$. Then there exists a unique extension $t\in\D'(\mathbb{R}^{d})$
of $t_{0}$, i.e. $t(f)=t_{0}(f)$ for all $f\in\D(\mathbb{R}^{d}\backslash\left\{ 0\right\} )$,
which has the same scaling degree, $\sd(t)=\sd(t_{0})$.
\item $d\leq\sd(t_{0})<\infty$. Then there exist several extensions $t\in\D'(\mathbb{R}^{d})$
with $\sd(t)=\sd(t_{0})$, which are uniquely determined by their
values on a finite set of test functions. 
\end{itemize}
\end{theorem}

The freedom in the possible extensions is described by the Stückelberg-\-Peter\-mann
renormalization group, see \citep{Brunetti2009} for a thorough discussion
of the topic. According to the theorem, given a particular extension
$t_{p}$, the most general solution for the extension $t$ of $t_{0}$
reads (cf. \citep{DuetschFredenhagen2004})\begin{equation}
t=t_{p}+\sum_{\left|a\right|\leq\rho(t_{0})}C_{a}\partial^{a}\delta\label{eq:Extension-MostGeneral}\end{equation}
where $\rho(t_{0}):=\sd(t_{0})-d$ denotes the degree of divergence
of $t_{0}$, $a\in\mathbb{N}^{d}$ is a multiindex and $C_{a}\in\mathbb{C}$.
Hence the freedom in the choice of an extension $t$ is governed by
the scaling degree of $t_{0}$.

Observe, however, that the existence of a renormalized time-ordered
product $E_{n}$, given by an extension (\ref{eq:Extension-MostGeneral}),
does not imply that the underlying theory is renormalizable in the
sense of power counting. For the sake of readability we quote here
the classification of renormalizable theories as it was given by Epstein
and Glaser \citep{Epstein1973}. Given an extension of $E_{n}$ of
$E_{n}^{0}$ exists for all orders $n$ of perturbation theory. A
theory is called \emph{renormalizable}, if there is a (finite) upper
bound for the degree of divergence of the $n$-fold time-ordered product,
$\rho(t_{0}^{n})$, which does not depend on the order $n$ of perturbation
theory. The theory is called \emph{unrenormalizable}, if there is
no such bound and it is called \emph{superrenormalizable}, if there
is a certain order $n_{0}$ above which the degree of divergence is
negative, i.e. the extensions $E_{n}$ are unique for $n>n_{0}$.

This closes the Euclidean version of the Epstein-Glaser induction
and at the same time shows that this induction is completely performable
without reference to the star-product structure of Minkowskian pQFT.

\section{Conclusion}

We have shown that the construction of Epstein and Glaser can be adapted
to the Euclidean case. On the one hand side, as asserted in the introduction
this shows that the construction of Epstein and Glaser can generally
be performed without the notion of a star-product. On the other hand
the formalism introduced here gives strong tools for the investigation
of the local properties of Euclidean QFT. Particularly interesting
is the investigation of the relation to other approaches to renormalization,
like for instance the BPHZ renormalization scheme. New results in
this direction have recently been gained by an investigation of certain
examples in momentum space \citep{Falk2009}. 

In principle it is possible to get back non-local objects like the
analytic Schwinger functions of EQFT from the formalism introduced
above. As in the Minkowskian setting of Algebraic Quantum Field Theory
one gets back the correlation functions by evaluating their corresponding
algebraic versions in the vacuum state. In the introduced setting
this is given by the evaluation at $\varphi=0$, \[
\omega_{0}:E_{n}(F_{1}\otimes\cdots\otimes F_{n})\mapsto E_{n}(F_{1}\otimes\cdots\otimes F_{n})\big|_{\varphi=0}\,.\]
Performing the adiabatic limit, provided it exists, gives back the
Schwinger functions of EQFT.

\begin{acknowledgments}
I want to thank K. Fredenhagen for inducing the investigation of the
topic. Many of the ideas leading to the result were suggested by him.
Furthermore I want to thank R. Brunetti, as well as P. Lauridsen Ribeiro
for useful suggestions and remarks. All three of them I'd like to
thank for critical reading of earlier versions of this manuscript.
I want to thank M. Dütsch for pointing out a gap in a previous version
of the proof of Proposition \ref{pro:Uniqueness-up-to-Diagonal}.
This work was financially supported by Deutsche Forschungsgemeinschaft
(GK 602).
\end{acknowledgments}
\appendix

\section{\label{sec:WF-delta}The Wave Front Set of the Dirac Delta Distribution}

Just as an example for the computation of the wave front set of a
given distribution we compute $\WF(\delta(x,y))$. The singular support
of $\delta$ is the diagonal $\Diag(\Em^{2})\equiv\left\{ \left(x,x\right)\in\Em^{2}\right\} $.
Following the definition \citep[Def. 8.1.2]{Hoermander2003} we are
interested in the directions $\left(k,k'\right)\in T_{\left(x,x\right)}^{*}\Em^{2}$
in which the Fourier transform of $\delta$, $\FT(\delta)\equiv\hat{\delta}$,
does not decrease rapidly.\begin{align*}
\hat{\delta}(k,k') & =\frac{1}{\left(2\pi\right)^{n}}\int dx\, dy\, e^{-i\left\langle k,x\right\rangle }e^{-i\left\langle k',y\right\rangle }\delta(x,y)\\
 & =\frac{1}{\left(2\pi\right)^{n}}\int dx\, e^{-i\left\langle k+k',x\right\rangle }=\delta(k+k')\end{align*}
where one can use Fourier's inversion formula to show the last equality.
Hence the Fouriertransform of $\delta(x,y)$ is rapidly decreasing
in all directions except $\left\{ \left(k,k'\right)\in T_{\left(x,x\right)}^{*}\Em^{2}:k+k'=0\right\} $.
The wave front set of $\delta$ is therefore given by\[
\WF(\delta)=\left\{ \left(x,k;x,k'\right)\in T^{*}\Em^{2}:\, k+k'=0\right\} \,.\]

\section{\label{sec:Graphs-and-Tadpoles}Combinatorics: Graphs and Tadpoles}

\subsection*{Graphs}

As is well known, the symmetry of the time-ordered product is conveniently
accounted for by writing its terms as (sums of) graphs. Using Cauchy's
product formula and the Leibniz rule one derives from the formal power
series (\ref{eq:T-Prod-SeriesExpansion}) the following expression
for the threefold Euclidean time-ordered product, whose addends have
a direct interpretation in terms of graphs, {\allowdisplaybreaks\begin{align*}
F\dE G\dE H\hspace{-15mm} & \hspace{15mm}=\sum_{n=0}^{\infty}\frac{\hbar^{n}}{n!}\sum_{m=k}^{n}\sum_{k=0}^{m}{n \choose m}{m \choose k}\left\langle F^{\left(k+m-k\right)}G_{\left(k\right)}^{\left(n-m\right)}H_{\left(n-m\right)\left(m-k\right)}\right\rangle \\
 & =FGH+\hbar\left(\left\langle FG^{\left(1\right)}H_{\left(1\right)}\right\rangle +\left\langle F^{\left(1\right)}GH_{\left(1\right)}\right\rangle +\left\langle F^{\left(1\right)}G_{\left(1\right)}H\right\rangle \right)\\
 & \quad+\hbar^{2}\left(\frac{1}{2}\left\langle FG^{\left(2\right)}H_{\left(2\right)}\right\rangle +\left\langle F^{\left(1\right)}G^{\left(1\right)}H_{\left(1\right)\left(1\right)}\right\rangle +\left\langle F^{\left(1\right)}G_{\left(1\right)}^{\left(1\right)}H_{\left(1\right)}\right\rangle \right.\\
 & \quad\qquad\qquad\qquad\qquad+\left.\frac{1}{2}\left\langle F^{\left(2\right)}GH_{\left(2\right)}\right\rangle +\left\langle F^{\left(2\right)}G_{\left(1\right)}H_{\left(1\right)}\right\rangle +\frac{1}{2}\left\langle F^{\left(2\right)}G_{\left(2\right)}H\right\rangle \right)\\
 & \quad+\cdots\\
 & =\mbox{\FGH}+\hbar\left(\mbox{\FGoneHone}+\mbox{\FoneGHone}+\mbox{\FoneGoneH}\right)\\
 & \quad+\hbar^{2}\left(\frac{1}{2}\,\mbox{\FGtwoHtwo}+\mbox{\FoneGoneHtwo}+\mbox{\FoneGoneoneHone}+\frac{1}{2}\,\mbox{\FtwoGHtwo}+\mbox{\FtwoGoneHone}+\frac{1}{2}\,\mbox{\FtwoGtwoH}\right)+\cdots\end{align*}
}

The symmetry factor $\Sym(\gamma)^{-1}=\frac{1}{n!}{n \choose m}{m \choose k}=\frac{1}{\left(n-m\right)!k!\left(m-k\right)!}$
of a given term is reflected in the graph as (the reciprocal of) the
product of the number of possible permutations of edges which join
the same vertices.%
\footnote{In the graph representing $\left\langle F^{\left(k+m-k\right)}G_{\left(k\right)}^{\left(n-m\right)}H_{\left(n-m\right)\left(m-k\right)}\right\rangle $
there are $k$ lines joining $F$ and $G$, $\left(n-m\right)$ lines
from $G$ to $H$, and $\left(m-k\right)$ edges connecting $H$ with
$F$.%
} This is what remains of the symmetry of the functional derivatives
after convolution with the fundamental solution $P$. The interpretation
in terms of graphs gives a straight-forward generalization to $n$-fold
products: \[
E_{n}(F_{1}\otimes\cdots\otimes F_{n})=\sum_{l=0}^{\infty}\hbar^{l}\sum_{\gamma\in\Gamma(n,l)}\frac{1}{\Sym(\gamma)}\,\gamma\,,\]
where $\Gamma(n,l)$ is the set of graphs with $n$ vertices and $l$
edges, in which each edge $e$ joins two different points, $\mathfrak{s}(e)\neq\mathfrak{r}(e)$
(no tadpoles). The graph with the interactions $\left\{ F_{1},\dots,F_{n}\right\} $
at $n$ vertices and $l_{i,j}$ edges between $F_{i}$ and $F_{j}$
corresponds to the term:\[
\left\langle \left(F_{1}\right)^{\left(l_{1,2}+l_{1,3}+\cdots+l_{1,n}\right)}\cdots\left(F_{k}\right)_{\left(l_{1,k}\right)\cdots\left(l_{k-1,k}\right)}^{\left(l_{k,k+1}+\cdots+l_{k,n}\right)}\cdots\left(F_{n}\right)_{\left(l_{1,n}\right)\cdots\left(l_{n-1,n}\right)}\right\rangle \,.\]
Notice that both the upper and the lower indices add up to the total
number $L$ of edges in the graph,\[
\sum_{i=1}^{n-1}\sum_{j=i+1}^{n}l_{i,j}=\sum_{j=2}^{n}\sum_{i=1}^{j-1}l_{i,j}=\sum_{i<j}l_{i,j}=L\,.\]

\subsection*{Tadpoles}

As asserted in the main part of the article, we want to prove here,
that there are no tadpole terms, i.e. graphs with at least one line
connecting a vertex with itself, in the graph-expansion of $F\dE G$
as defined in (\ref{eq:T-Prod-Operations}). By doing so, we give
the justification for formula (\ref{eq:T-Prod-SeriesExpansion}).

\begin{proposition}There are no tadpole terms contributing to ${F\dE G:=}$
 ${T_{\mathrm{E}}\circ M\circ\left(T_{\mathrm{E}}^{-1}(F)\otimes T_{\mathrm{E}}^{-1}(G)\right)}$,
that is:\[
F\dE G=\sum_{n=0}^{\infty}\frac{\hbar^{n}}{n!}\left\langle F^{\left(n\right)}G_{\left(n\right)}\right\rangle \,.\]
\end{proposition}

\begin{proof}The Euclidean time ordering operator, as well as the
corresponding operator(s) in pAQFT \citep{Brunetti2009} are induced
by second order functional differential operators, cf. Eq. (\ref{eq:T-Operator}),\[
\Gamma=\frac{1}{2}\int dx\, dy\, P(x,y)\,\frac{\delta^{2}}{\delta\varphi(x)\,\delta\varphi(y)}\,.\]
As differential operators on $\mathcal{F}(\Em)$ they fulfill the
Leibniz rule, which in turn may be written as a co-product rule:%
\footnote{see also \citep{BrouderCargese2009}%
}\[
\Gamma(F\cdot G)=M\circ\left(\Delta\Gamma\right)(F\otimes G)\,,\quad\Delta\Gamma=\Gamma\otimes1+1\otimes\Gamma+\Gamma',\]
where \[
\Gamma'(F\otimes G)=\int dx\, dy\, P(x,y)\frac{\delta F}{\delta\varphi(x)}\otimes\frac{\delta G}{\delta\varphi(y)}\,.\]

The time-ordered product hence is given by, cf. Eq. (\ref{eq:T-Prod-Operations}),\[
F\dE G=e^{\hbar\Gamma}\circ M\circ\left(e^{-\hbar\Gamma}F\otimes e^{-\hbar\Gamma}G\right)\,.\]
Applying the Leibniz rule and using the functional identity for the
exponential ($e^{A}e^{B}=e^{A+B}$), which holds due to commutativity
and associativity of the product of differential operators, leads
to: \begin{align*}
F\dE G & =M\circ e^{\hbar\Delta\Gamma}\circ\left(e^{-\hbar\Gamma}F\otimes e^{-\hbar\Gamma}G\right)\\
 & =M\circ e^{\hbar\Gamma'}\left(e^{\hbar\Gamma}e^{-\hbar\Gamma}F\otimes e^{\hbar\Gamma}e^{-\hbar\Gamma}G\right)\\
 & =M\circ e^{\hbar\Gamma'}\left(F\otimes G\right)\,.\end{align*}
Hence the result stated before:\[
F\dE G=\sum_{n=0}^{\infty}\frac{\hbar^{n}}{n!}\left\langle F^{\left(n\right)},P^{\otimes n}G^{\left(n\right)}\right\rangle \,.\]
\end{proof}

\section{\label{app:Local-Functionals}Local Functionals}

The definition of a local functional in \citep{Brunetti2009} differs
from the one given in this article. Hence, in order to be able to
apply their results on local functionals in our context, we have to
make sure that the functionals fulfilling the conditions of Definition
\ref{def:Local-Functional} are a subset of the set of local functionals
in the sense of \citep[Section 3.2]{Brunetti2009}. For this it suffices
to show that the support property {[}LF-1] implies the additivity
condition of \citep{Brunetti2009}. Together with Lemma 3.1 of \citep{Brunetti2009}
this proves equivalence of both definitions. The argument is taken
from \citep{BrunettiFredenhagenRibeiro2009}.

\begin{proof}Let $F$ be a smooth functional fulfilling\begin{equation}
\frac{\delta^{2}F}{\delta\varphi(x)\delta\varphi(y)}=0\qquad\mbox{for}\, x\neq y\,,\label{eq:LF-equivalence1}\end{equation}
we have to show that\begin{equation}
\supp(\varphi)\cap\supp(\chi)=\emptyset\quad\mbox{implies}\quad\forall\psi:\; F(\varphi+\psi+\chi)=F(\varphi+\psi)-F(\chi)+F(\psi+\varphi)\,.\label{eq:LF-equivalence2}\end{equation}

We have $\forall\psi$:\begin{equation}
\frac{\partial^{2}}{\partial\lambda\,\partial\mu}F(\lambda\varphi+\psi+\mu\chi)=\int dx\, dy\, F^{\left(2\right)}(\lambda\varphi+\psi+\mu\chi)(x,y)\,\varphi(x)\,\chi(y)\,,\label{eq:LF-equivalence3}\end{equation}
where due to (\ref{eq:LF-equivalence1}) the domain of integration
can be restricted to the diagonal $\left\{ \left(x,y\right):x=y\right\} $.
And since $\forall x$: $\varphi(x)\chi(x)=0$ due to the assumption
in (\ref{eq:LF-equivalence2}), we have that the integral on the right
hand side of (\ref{eq:LF-equivalence3}) vanishes, i.e.\[
\frac{\partial^{2}}{\partial\lambda\,\partial\mu}F(\lambda\varphi+\psi+\mu\chi)\equiv0\,.\]
Integration with respect to $\lambda$ gives\begin{align*}
\frac{\partial}{\partial\mu}F(\varphi+\psi+\mu\chi) & =\int_{0}^{1}d\lambda\,\frac{\partial^{2}}{\partial\lambda\,\partial\mu}F(\lambda\varphi+\psi+\mu\chi)+\frac{\partial}{\partial\mu}F(\psi+\mu\chi)\\
 & =\frac{\partial}{\partial\mu}F(\psi+\mu\chi)\,,\end{align*}
and integrating another time with respect to $\mu$ gives the desired
result:\begin{align*}
F(\varphi+\psi+\chi) & =F(\varphi+\psi)+\int_{0}^{1}d\mu\,\frac{\partial}{\partial\mu}F(\psi+\mu\chi)\\
 & =F(\varphi+\psi)-F(\psi)+F(\psi+\chi)\,.\end{align*}
\end{proof}

\end{fmffile}

\bibliographystyle{/scratch/keller/Literatur/kaialpha}
\bibliography{/scratch/keller/Literatur/Literatur_Physik}

\begin{thebibliography}{BOR02}

\bibitem[BDF09]{Brunetti2009}
R.~Brunetti, M.~D{\"u}tsch, and K.~Fredenhagen.
\newblock {P}erturbative {A}lgebraic {Q}uantum {F}ield {T}heory and the
  {R}enormalization {G}roups.
\newblock 2009.
\newblock
\newblock \href{http://arxiv.org/abs/0901.2038}{arXiv:0901.2038}.

\bibitem[BF00]{Brunetti2000}
R.~Brunetti and K.~Fredenhagen.
\newblock {M}icrolocal {A}nalysis and {I}nteracting {Q}uantum {F}ield
  {T}heories: {R}enormalization on {P}hysical {B}ackgrounds.
\newblock {\em Commun. Math. Phys.}, 208(3):623--661,
\newblock 2000.

\bibitem[BFR09]{BrunettiFredenhagenRibeiro2009}
R.~Brunetti, K.~Fredenhagen, and P.~Lauridsen Ribeiro.
\newblock 2009.
\newblock
\newblock forthcoming paper.

\bibitem[BG96]{BolliniGiambiagi1996}
C.~G. Bollini and J.~J. Giambiagi.
\newblock {D}imensional regularization in configuration space.
\newblock {\em Phys. Rev. D}, 53(10):5761,
\newblock May 1996.

\bibitem[BOR02]{Buchholz2002}
D.~Buchholz, I.~Ojima, and H.~Roos.
\newblock {T}hermodynamic {P}roperties of {N}on-{E}quilibrium {S}tates in
  {Q}uantum {F}ield {T}heory.
\newblock {\em Ann. Physics}, 297(2):219--242,
\newblock 2002.

\bibitem[BP57]{Bogoliubow1957}
N.~N. Bogoliubow and O.~S. Parasiuk.
\newblock {\"U}ber die {M}ultiplikation der {K}ausalfunktionen in der
  {Q}uantentheorie der {F}elder.
\newblock {\em Acta Mathematica}, 97(1-4):227--266,
\newblock 1957.

\bibitem[Bro09]{BrouderCargese2009}
C.~Brouder.
\newblock {H}opf algebraic structures of quantum field and many-body theories.
\newblock Lecture given at the ACQFT conference in Carg{\`e}se, Corsica, 2009.
\newblock \url{http://www.lmcp.jussieu.fr/~brouder/talk.html}.

\bibitem[BS59]{BogoliubovShirkov1959}
N.~N. Bogoliubov and D.~V. Shirkov.
\newblock {\em {I}ntroduction to the {T}heory of {Q}uantized {F}ields}.
\newblock Interscience Publishers, 1959.

\bibitem[DF01]{Dutsch2000}
M.~D{\"u}tsch and K.~Fredenhagen.
\newblock {P}erturbative {A}lgebraic {F}ield {T}heory, and {D}eformation
  {Q}uantization.
\newblock In R.~Longo, editor, {\em Mathematical Physics in Mathematics and
  Physics: Quantum and Operator Algebraic Aspects}, volume~30 of {\em Fields
  Institute Communications}, Providence, RI, 2001. AMS.
\newblock Proceedings of the conference held in Siena on June 20-24, 2000;
  \href{http://arxiv.org/abs/hep-th/0101079}{arXiv:hep-th/0101079}.

\bibitem[DF04]{DuetschFredenhagen2004}
M.~D{\"u}tsch and K.~Fredenhagen.
\newblock {C}ausal {P}erturbation {T}heory in {T}erms of {R}etarded {P}roducts,
  and a {P}roof of the {A}ction {W}ard {I}dentity.
\newblock {\em Rev. Math. Phys.}, 16(10):1291--1348,
\newblock 2004.

\bibitem[EE79]{EckmannEpstein1979}
J.~P. Eckmann and H.~Epstein.
\newblock {T}ime-ordered products and {S}chwinger functions.
\newblock {\em Commun. Math. Phys.}, 64(2):95--130,
\newblock June 1979.

\bibitem[EG73]{Epstein1973}
H.~Epstein and V.~Glaser.
\newblock {T}he {R}ole of {L}ocality in {P}erturbation {T}heory.
\newblock {\em Ann. Inst. Henri Poincar{\'e}}, 19(3):211--295,
\newblock 1973.

\bibitem[FHS09]{Falk2009}
S.~Falk, R.~H{\"a}u{\ss}ling, and F.~Scheck.
\newblock {R}enormalization in {Q}uantum {F}ield {T}heory: {A}n {I}mproved
  {R}igorous {M}ethod.
\newblock 2009.
\newblock
\newblock \href{http://arxiv.org/abs/0901.2252}{arXiv:0901.2252}.

\bibitem[GS08]{GrigoreScharf2008}
D.R. Grigore and G.~Scharf.
\newblock {N}o-{G}o {R}esult for {S}upersymmetric {G}auge {T}heories in the
  {C}ausal {A}pproach.
\newblock {\em Annalen der Physik}, 17(11):864--880,
\newblock 2008.

\bibitem[Hep66]{Hepp1966}
K.~Hepp.
\newblock {P}roof of the {B}ogoliubov-{P}arasiuk {T}heorem on
  {R}enormalization.
\newblock {\em Commun. Math. Phys.}, 2(4):301--326,
\newblock 1966.

\bibitem[HH02]{Hirshfeld2002}
A.~C. Hirshfeld and P.~Henselder.
\newblock {S}tar {P}roducts and {P}erturbative {Q}uantum {F}ield {T}heory.
\newblock {\em Ann. Physics}, 298(2):382--393,
\newblock 2002.

\bibitem[H{\"o}r03]{Hoermander2003}
L.~H{\"o}rmander.
\newblock {\em {T}he {A}nalysis of {L}inear {P}artial {D}ifferential
  {O}perators {I}: {D}istribution {T}heory and {F}ourier {A}nalysis}.
\newblock Classics in Mathematics. Springer, 2003.

\bibitem[HW01]{HollandsWald2001}
S.~Hollands and R.~M. Wald.
\newblock {L}ocal {W}ick {P}olynomials and {T}ime {O}rdered {P}roducts of
  {Q}uantum {F}ields in {C}urved {S}pacetime.
\newblock {\em Commun. Math. Phys.}, 223(2):289--326,
\newblock 2001.

\bibitem[OS73]{OsterwalderSchrader1973}
K.~Osterwalder and R.~Schrader.
\newblock {A}xioms for {E}uclidean {G}reen's functions.
\newblock {\em Commun. Math. Phys.}, 31(2):83--112,
\newblock 1973.

\bibitem[OS75]{OsterwalderSchrader1975}
K.~Osterwalder and R.~Schrader.
\newblock {A}xioms for {E}uclidean {G}reen's functions {II}.
\newblock {\em Commun. Math. Phys.}, 42(3):281--305,
\newblock 1975.

\bibitem[Roe94]{Roepstorff1994}
G.~Roepstorff.
\newblock {\em {P}ath {I}ntegral {A}pproach to {Q}uantum {P}hysics}.
\newblock Springer, 1994.

\bibitem[RS80]{ReedSimon1980Vol1}
M.~Reed and B.~Simon.
\newblock {\em {F}unctional {A}nalysis}, volume~1 of {\em Methods of Modern
  Mathematical Physics}.
\newblock Academic Press, 1980.

\bibitem[Sal99]{Salmhofer1999}
M.~Salmhofer.
\newblock {\em {R}enormalization - {A}n {I}ntroduction}.
\newblock Springer, 1999.

\bibitem[Sch59]{Schwinger1959}
J.~Schwinger.
\newblock {Euclidean Quantum Electrodynamics}.
\newblock {\em Phys. Rev.}, 115:721--731,
\newblock 1959.

\bibitem[Sib08]{Sibold2008}
K.~Sibold.
\newblock {C}omment to \textquoteleft {N}o-{G}o {R}esult for {S}upersymmetric
  {G}auge {T}heories in the {C}ausal {A}pproach\textquoteright, by {D}. {R}.
  {G}rigore and {G}. {S}charf (previous article: {A}nn. {P}hys. ({B}erlin) 17,
  864 (2008)).
\newblock {\em Annalen der Physik}, 17(11):881--884,
\newblock 2008.

\bibitem[SR50]{StueckelbergRivier1950}
E.~C.~G. St{\"u}ckelberg and D.~Rivier.
\newblock {A} propos des divergences en th{\'e}orie des champs quantifi{\'e}s.
\newblock {\em Helv. Phys. Acta}, 23 (Suppl. III):236--239,
\newblock 1950.

\bibitem[Ste71]{Steinmann1971}
O.~Steinmann.
\newblock {\em {P}erturbation {E}xpansions in {A}xiomatic {F}ield {T}heory},
  volume~11 of {\em Lecture Notes in Physics}.
\newblock Springer, Berlin and Heidelberg, 1971.

\bibitem[Sto06]{Stora2006}
R.~Stora.
\newblock {C}ausalit{\'e} et {G}roupes de {R}enormalisation {P}erturbatifs.
\newblock unpublished lecture notes, Ecole de Physique Th{\'e}orique de Jijel,
  Alg{\'e}rie, 2006.

\bibitem[SW64]{StreaterWightman1964}
R.~F. Streater and A.~S. Wightman.
\newblock {\em {PCT}, {S}pin and {S}tatistics, and all that}.
\newblock W.A. Benjamin, Inc., Amsterdam and New York, 1964.

\bibitem[Sym69]{Symanzik1969}
K.~Symanzik.
\newblock {E}uclidean {Q}uantum {F}ield {T}heory.
\newblock In R.~Jost, editor, {\em Local Quantum Theory}, number Course XLV in
  Proceedings of the International School of Physics 'Enrico Fermi', pages
  152--226, New York and London, 1969. Academic Press.

\bibitem[Zim69]{Zimmermann1969}
W.~Zimmermann.
\newblock {C}onvergence of {B}ogoliubov's {M}ethod of {R}enormalization in
  {M}omentum {S}pace.
\newblock {\em Commun. Math. Phys.}, 15(3):208--234,
\newblock 1969.

\end{thebibliography}

\end{document}